\documentclass[11pt]{article}
\usepackage{pset}
\title{Basis Collapse for Holographic Algorithms Over All Domain Sizes}
\author{\begin{tabular}{*{3}{>{\centering}p{.31\textwidth}}}
\large Sitan Chen\thanks{Harvard University Department of Mathematics \& Harvard John A. Paulson School of Engineering and Applied Sciences} \tabularnewline
{\footnotesize \url{sitanchen@college.harvard.edu}}
\end{tabular}}

\begin{document}
\maketitle

\setcounter{thm}{0}

\begin{abstract}
\label{abstract}

The theory of holographic algorithms introduced by Valiant represents a novel approach to achieving polynomial-time algorithms for seemingly intractable counting problems via a reduction to counting planar perfect matchings and a linear change of basis. Two fundamental parameters in holographic algorithms are the \emph{domain size} and the \emph{basis size}. Roughly, the domain size is the range of colors involved in the counting problem at hand (e.g. counting graph $k$-colorings is a problem over domain size $k$), while the basis size $\ell$ captures the dimensionality of the representation of those colors. A major open problem has been: for a given $k$, what is the smallest $\ell$ for which any holographic algorithm for a problem over domain size $k$ ``collapses to'' (can be simulated by) a holographic algorithm with basis size $\ell$? Cai and Lu showed in 2008 that over domain size 2, basis size 1 suffices, opening the door to an extensive line of work on the structural theory of holographic algorithms over the Boolean domain. Cai and Fu later showed for signatures of full rank that over domain sizes 3 and 4, basis sizes 1 and 2, respectively, suffice, and they conjectured that over domain size $k$ there is a collapse to basis size $\lfloor\log_2 k\rfloor$. In this work, we resolve this conjecture in the affirmative for signatures of full rank for all $k$.
\end{abstract}

\section{Introduction}
\label{sec:intro}

\subsection{Matchgates and Holographic Algorithms}

In \cite{valiant1,valiant2}, Valiant introduced the notion of matchgate computation as a method for classically simulating certain quantum gates in polynomial time and asked whether matchgates could be used to derive other polynomial-time algorithms. Given a counting problem, one might hope that its local combinatorial constraints can be encoded by graph fragments $G_i$ so that there exists a bijection between the number of solutions to that counting problem and the number of perfect matchings of an amalgamation $\Omega$ of those graph fragments. The celebrated Fisher-Kasteleyn-Temperley algorithm \cite{fk1,temp} then computes the latter in polynomial time if $\Omega$ is planar.

The drawback of the above approach is that if one seeks a one-to-one reduction to counting perfect matchings in a planar graph, a matchgate-based solution might not necesssarily exist because the range of encodable local constraints is too limited. Valiant's insight in \cite{valiant3,valiant4} was to extend this range by looking instead for a \emph{many-to-many} reduction. In his new framework, multiple strands of computation get combined in a ``holographic'' mixture with exponential, custom-built cancellations specified by a choice of basis vectors to produce the final answer. With this change-of-basis approach, Valiant \cite{valiant3,valiant4,valiant5,valiant6} found polynomial-time solutions to a number of counting problems, minor variants of which are known to be intractable, and noted that the only criterion for their existence was the solvability of certain finite systems of polynomial equations. As a notable example, whereas counting the number of satisfying assignments to planar, read-twice, monotone 3-CNFs is $\#\P$-complete and even the same problem modulo 2 is $\NP$-hard under randomized reductions, the same problem modulo 7, known as $\#_7$Pl-Rtw-Mon-3CNF, has a polynomial-time holographic solution \cite{valiant3}.

Given the dearth of strong unconditional circuit lower bounds to support the prevailing belief that $\P\neq\NP$, we might question whether our current inability to devise polynomial-time solutions to $\NP$-complete problems is sufficient to justify such a belief. After all, it would also have justified the belief that $\#_7$Pl-Rtw-Mon-3CNF is intractable. This suggests that to arrive upon the desired separation of $\P$ and $\NP$, we need a better understanding of the possibilities of polynomial-time computation. For this reason, determining the ultimate capabilities of holographic algorithms appears to be a crucial step.

To this end, one major direction has been to study the kinds of local constraints that matchgates can encode. Cai and his collaborators initiated a systematic study of the structural theory of holographic algorithms over the Boolean domain \cite{caichoudhary1,caichoudhary2,caichoudhary3,caichoudharylu,caifu,caigoren,cailu1,cailu2,cailu3,cailu4}, compiling what amounts to a catalogue of such constructions that turns the process of finding basic holographic reductions into something essentially algorithmic. For example, as a corollary to their results in \cite{cailu1}, they noted that the modulus 7 appears in $\#_7$Pl-Rtw-Mon-3CNF because it is a Mersenne prime and that more generally, there is a polynomial-time holographic algorithm for $\#_{2^{k}-1}$Pl-Rtw-Mon-$k$CNF.

\subsection{Basis Collapse Theorems}

Another direction of study concerns better understanding the power afforded by the change of basis, specifically the relationship between two important parameters of holographic algorithms: the dimension and number of basis vectors. Because holographic algorithms reduce counting problems to counting perfect matchings, the dimension of the $k$ basis vectors specifying the abovementioned ``holographic mixture'' must be of the form $2^{\ell}$. We say that $\ell$ is the \emph{basis size} and $k$ is the domain size. The domain size should be interpreted as the range over which variables in the counting problem can take values, so for instance, problems related to counting certain $k$-colorings in a graph are problems over domain size $k$. The basis size should then be interpreted roughly as the number of bits needed to encode each of these $k$ colors.

The first holographic algorithms studied were on domain size 2, dealing with counting problems involving matchings, 2-colorings, graph bipartitions, and Boolean satisfying assignments, and used bases of size 1. The original holographic algorithm for $\#_7$Pl-Rtw-Mon-3CNF was the first to use a basis of size 2, and at the time it may have been tempting to believe that increasing basis size adds power to the holographic approach. In \cite{cailu2}, Cai and Lu showed to the contrary that on domain size 2, any nontrivial holographic algorithm with basis of size $\ell\ge 2$ can be simulated by one with basis of size 1. In other words, for all $\ell$, the class of domain size 2 problems solvable with basis size $\ell$ \emph{collapses} to the class solvable with basis size 1.

In \cite{caifu}, Cai and Fu further showed that holographic algorithms on domain sizes 3 and 4 using at least one full-rank signature collapse to basis sizes 1 and 2 respectively. They conjectured that for domain size $k$, we get a collapse to basis size $\lfloor\log_2 k\rfloor$ and suggested a heuristic explanation that for domain size $k = 2^K$, we only need $\log_2 k$ bits to encode each of the $k$ colors.

\subsection{Our Results and Techniques}

We prove Cai and Fu's conjecture in the affirmative for all domain sizes.

\begin{thm}
	Any holographic algorithm on domain size $k$ using at least one matchgate with signature of rank $k$ can be simulated by a holographic algorithm with basis size $\lfloor\log_2 k\rfloor$.
	\label{thm:main}
\end{thm}

As Cai and Fu noted in \cite{caifu}, their ``information-theoretic'' explanation for the collapse theorems on domain sizes up to 4 is insufficient to explain why holographic algorithms on domain size 3 collapse to basis size 1 and not just 2.

To prove their collapse theorem on domain size 3, Cai and Fu actually showed that the bases of holographic algorithms on domain size 3 which use at least one full-rank signature must be of rank 2 rather than 3. They then invoked the main result of Fu and Yang \cite{fuyang} which reduces holographic algorithms on domain size $k$ with bases of rank 2 to ones on domain size 2.

Our key observation is that this phenomenon occurs at a much larger scale. As a bit of informal background, the \emph{standard signature} of a matchgate $G$ is a vector encoding the perfect matching properties of $G$. By indexing appropriately, we can regard the standard signature as a matrix $\Gamma$. The entries of this matrix are known \cite{caigoren} to satisfy a collection of quadratic polynomial identities called the Matchgate Identities (MGIs), and by using these identities together with some multilinear algebra, we prove the following:

\begin{thm}[Rank Rigidity Theorem]
	The rank of the standard signature matrix $\Gamma$ is always a power of two.\label{thm:rigid_informal}
\end{thm}

We can then conclude that the basis of a nontrivial holographic algorithm on domain size $k$ must be of rank $2^{\ell}$, where $\ell$ is the largest integer for which $2^{\ell}\le k$. With this step, together with a generalization of Fu and Yang's result to bases of rank $k$, we show it is enough to prove a collapse theorem for holographic algorithms on domain sizes that are powers of two. Cai and Fu \cite{caifu} achieved such a collapse theorem for domain size 4 by proving that 1) any standard signature of rank 4 contains a full-rank $4\times 4$ submatrix whose entries have indices are ``close'' in Hamming distance, 2) full-rank $4\times 2^{2n-2}$ standard signatures have right inverses that are also standard signatures.

For 1), the proof in \cite{caifu} used algebraic techniques involving the matchgate identities, but these methods seem to work only for domain size 4. We instead show that the required generalization of 1) to arbitrary domain sizes almost trivially follows from the rank rigidity theorem and the MGIs. Roughly, we prove the following:

\begin{thm}[Cluster Existence - informal]
		If $\Gamma$ is a $2^{\ell}\times 2^{(n-1)\ell}$ matrix of rank at least $k$ realizable as the standard signature of some matchgate, then there exists a 
		$2^{\lceil\log_2 k\rceil}\times 2^{\lceil\log_2 k\rceil}$ submatrix of full rank in $\Gamma$ whose column (resp. row) indices differ in at most $\lceil\log_2 k\rceil$ bits.
		\label{thm:sub_informal}
\end{thm}

For 2), the proof in \cite{caifu} nonconstructively verifies that the set of all invertible $4\times 4$ matrices satisfying the matchgate identities up to sign forms a group under multiplication. $4\times 4$ matrices are easy to handle because there is only one nontrivial MGI in this case. Rather than generalizing this approach, we note that Li and Xia \cite{groupprop} proved a very similar but more general result under a different framework of matchgate computation known as character theory, showing that the set of all invertible $2^K\times 2^K$ matrices realizable as matchgate characters forms a group under multiplication for all $K$. It turns out their technique carries over with minor modifications into the framework of signature theory that we consider, and we use it to show the following:

\begin{thm}[Group Property - informal]
	If $G$ is a generator matchgate with $2^K\times 2^{(n-1)K}$ standard signature $\underline{G}$ of rank $2^K$, then there exists a recognizer matchgate with $2^{K(n-1)}\times 2^K$ standard signature $\underline{R}$ such that $\underline{G}\cdot\underline{R} = I_{2^K}$, where $I_{2^K}$ denotes the $2^K\times 2^K$ identity matrix.\label{thm:groupprop_informal}
\end{thm}

Our general collapse theorem then follows from our generalizations of 1) and 2). This result gives a way forward for the development of the structural theory of holographic algorithms on higher domain sizes in the same vein as Cai et al.'s work on domain size 2. In \cite{valiant6}, Valiant gave examples of holographic algorithms on domain size 3, but holographic algorithms on higher domain sizes have yet to be explored. Our result shows that for domain size $k$, we can focus on understanding changes of basis in $\mathcal{M}_{2^{\lfloor\log_2 k\rfloor}\times k}(\mathds{C})$ rather than over an infinite set of dimensions, just as the collapse theorem of Cai and Lu \cite{cailu2} showed that on the Boolean domain, they could focus on understanding changes of basis in $\GL_{2}(\mathds{C})$.

\subsection{Organization}

In Section~\ref{sec:prelims}, we introduce the necessary notation and background, including a brief overview of holographic algorithms. In Section~\ref{sec:mgis}, we state the matchgate identities and some of their key consequences. In Section~\ref{sec:rigidclust}, we prove Theorems~\ref{thm:rigid_informal} and \ref{thm:sub_informal}. In Section~\ref{sec:collapse-group}, we then prove Theorem~\ref{thm:groupprop_informal}. In Section~\ref{sec:fuyang}, we generalize the main result of Fu and Yang \cite{fuyang} to bases of rank $k$ and reduce proving a collapse theorem on all domain sizes to proving one on domain sizes equal to powers of two. Finally, in Section~\ref{sec:collapse-final}, we prove the desired collapse theorem, Theorem 1.1, on domain sizes equal to powers of two by invoking the results from Section~\ref{sec:rigidclust}.

\section{Preliminaries}
\label{sec:prelims}

\subsection{Background}

Denote the Hamming weight of string $\alpha$ by $\wt(\alpha)$, and define the parity of $\alpha$ to be the parity of $\wt(\alpha)$. Given $1\le i\le m$, define $e_i\in\{0,1\}^m$ to be the bitstring with a single nonzero bit in position $i$. The parameter $m$ is implicit, and when this notation is used, $m$ will be clear from the context. Denote by $1^m$ the length-$m$ bitstring consisting solely of 1's.

We review some basic definitions and results about holographic algorithms. For a comprehensive introduction to this subject, see \cite{valiant4}.

\begin{defn}
	A \emph{matchgate} $\Gamma = (G,X,Y)$ is defined by a planar embedding of a planar graph $G = (V,E,W)$, \emph{input nodes} $X\subseteq V$, and \emph{output nodes} $Y\subseteq V$, where $X\cap Y = \emptyset$. We refer to $X\cup Y$ as the \emph{external nodes} of $\Gamma$.

	We say that $\Gamma$ has \emph{arity} $|X|+|Y|$. In the planar embedding of $G$, the input and output nodes are arranged such that if one travels counterclockwise around the outer face of $G$, one encounters first the input nodes labeled 1,2,...,$|X|$ and then the output nodes $|Y|$,...,2,1.

	If $\Gamma$ has exclusively output (resp. input) nodes, we say that $\Gamma$ is a generator (resp. recognizer). Otherwise, we say that $\Gamma$ is an \emph{$|X|$-input, $|Y|$-output transducer}.
\end{defn}

\begin{defn}
A \emph{basis} of size $\ell$ on domain size $k$ is a $2^{\ell}\times k$ matrix $M = (a^{\alpha}_i)$, where rows and columns are indexed by $\alpha\in\{0,1\}^{\ell}$ and $i\in[k]$ respectively.
\end{defn}


\begin{defn}
The \emph{standard signature} of a matchgate $\Gamma$ of arity $n\ell$ is a vector of dimension $2^{n \ell}$ which will be denoted by $\underline{\Gamma}$, where for $\alpha_i\in\{0,1\}^{\ell}$, $\underline{\Gamma}^{\alpha_1\cdots\alpha_n}$ denotes the entry of $\underline{\Gamma}$ indexed by $\alpha_1\circ\cdots\circ\alpha_n$. If $Z$ is the subset of the external nodes of $\Gamma$ for which $\alpha_1\circ\cdots\circ\alpha_n$ is the indicator string, then $$\underline{\Gamma}^{\alpha_1\cdots\alpha_n} = \PerfMatch(\Gamma\backslash Z).$$ Here, if $A = (A_{ij})$ is the adjacency matrix of $\Gamma$, $\PerfMatch$ is the polynomial in the entries of $A$ defined by $$\PerfMatch(A) = \sum_{M}\prod_{(i,j)\in M}A_{ij},$$ with the sum taken over the set $M$ of all perfect matchings of $\Gamma.$
\end{defn}

The following lemma follows from the definition of standard signatures.

\begin{lem}
	Suppose $\underline{R}$ is the standard signature of a recognizer of arity $n\ell$ and $T$ the standard signature of a transducer with $s$ inputs and $\ell$ outputs. Then $\underline{R}' = \underline{R}T^{\tensor n}$ is the standard signature of a recognizer matchgate of arity $ns$.\label{lem:triv_transducer}
\end{lem}

\begin{defn}
A column (resp. row) vector of dimension $k^n$ is said to be a \emph{generator (resp. recognizer) signature realizable over a basis $M$} if there exists a generator (resp. recognizer) matchgate $\Gamma$ satisfying $M^{\tensor n}G = \underline{G}$ (resp. $\underline{R}M^{\tensor n} = R$). We say that a collection of recognizer and generator signatures $R_1,...,R_a,G_1,...,G_b$ is \emph{simultaneously realizable} if they are realizable over a common basis $M$.
\end{defn}

In particular, if $M$ is square, the signature of a matchgate with respect to the standard basis is the standard signature. Also note that in terms of coordinates, we have that $$\underline{G}^{\alpha_1\cdots\alpha_n} = \sum_{j_1,...,j_n\in[k]}G^{j_1\cdots j_n}a^{\alpha_1}_{j_1}\cdots a^{\alpha_n}_{j_n}$$ and $$R_{j_1\cdots j_n} = \sum_{\alpha_1,...,\alpha_n\in\{0,1\}^{\ell}}\underline{R}_{\alpha_1\cdots\alpha_n} a^{\alpha_1}_{j_1}\cdots a^{\alpha_n}_{j_n}.$$

\begin{defn}
	A \emph{matchgrid} $\Omega = (G,R,W)$ is a weighted planar graph consisting of a set of $g$ generators $G = \{G_1,...,G_g\}$, a set of $r$ recognizers $R = \{R_1,...,R_r\}$, and a set of $w$ wires $W = \{W_1,...,W_w\}$, each of which has weight 1 and connects the output node of a generator to the input node of a recognizer so that every input and output node among the matchgates in $G\cup R$ lies on exactly one wire.
\end{defn}

\begin{defn}
	Suppose $\Omega = (G,R,W)$ is a matchgrid with $g$ generators, $r$ recognizers, and $w$ wires, and let $M$ be a basis for $\Omega$. Define the \emph{Holant} to be the following quantity: $$\Holant(\Omega) = \sum_{z\in[k]^w}\left(\prod^g_{i=1}G^{y_i}_i\prod^r_{j=1}R^{x_j}_j\right).$$ Here, $z = y_1\circ\cdots\circ y_g = x_1\circ\cdots\circ x_r$ such that $y_i\in [k]^{|Y_i|}$ and $x_j\in [k]^{|X_j|}$ for $Y_i$ the output nodes of $G_i$ and $X_j$ the input nodes of $R_j$, and $G_i$ and $R_j$ denote the signatures of their respective matchgates under basis $M$.
\end{defn}

Valiant's Holant theorem states the following.

\begin{thm}[Theorem 4.1, \cite{valiant4}]
	If $\Omega$ is a matchgrid over a basis $M$, then $\Holant(\Omega) = \PerfMatch(\Omega)$.
\end{thm}

As the Fisher-Kasteleyn-Temperley algorithm \cite{fk1,temp} can compute the number of perfect matchings of a planar graph in polynomial time, $\Holant(\Omega)$ can be computed in polynomial time as long as $\Omega$ is planar.

\subsection{Matrix Form of Signatures}

It will be convenient to regard signatures not as vectors but as matrices.

\begin{defn}
For generator signature $G$, the $t$-th matrix form $G(t)$ ($1\le t\le n$) is a $k\times k^{n-1}$ matrix where the rows are indexed by $1\le j_t\le k$ and the columns are indexed by $j_1\cdots j_{t-1}j_{t+1}\cdots j_n$ in lexicographic order.
\end{defn}

\begin{defn}
For recognizer signature $R$, the $t$-th matrix form $R(t)$ ($1\le t\le n$) is a $k^{n-1}\times k$ matrix where the rows are indexed by $j_1\cdots j_{t-1}j_{t+1}\cdots j_n$ in lexicographic order and the columns are indexed by $1\le j_t\le k$.
\end{defn}

We would also like to regard standard signatures as matrices; if basis $M$ is square, the following definitions are special cases of the above.

\begin{defn}
For standard signature $\underline{G}$, the $t$-th matrix form $\underline{G}(t)$ ($1\le t\le n$) is a $2^{\ell}\times 2^{(n-1)\ell}$ matrix where the rows are indexed by $\alpha_t$ and the columns are indexed by $\alpha_1\cdots\alpha_{t-1}\alpha_{t+1}\cdots\alpha_n$.
\end{defn}

\begin{defn}
For standard signature $\underline{R}$, the $t$-th matrix form $\underline{R}(t)$ ($1\le t\le n$) is a $2^{(n-1)\ell}\times 2^{\ell}$ matrix where the rows are indexed by $\alpha_1\cdots\alpha_{t-1}\alpha_{t+1}\cdots\alpha_n$ and the columns are indexed by $\alpha_t$.
\end{defn}

One can check that $\underline{G}(t)$ and $G(t)$, and $\underline{R}(t)$ and $R(t)$, are related as follows.

\begin{lem}
	If $\underline{G} = M^{\tensor n}G$, then $$\underline{G}(t) = MG(t)(M^T)^{\tensor(n-1)}.$$\label{lem:relunderline}
\end{lem}

\begin{lem}
	If $R = \underline{R}M^{\tensor n}$, then $$R(t) = (M^T)^{\tensor(n-1)}\underline{R}(t)M.$$\label{lem:relunderline2}
\end{lem}

We will denote by $\underline{G}(t)^{\sigma}$ the row vector indexed by $\sigma$, $\underline{G}(t)_{\zeta}$ the column vector indexed by $\zeta$, and $\underline{G}(t)^{\sigma}_{\zeta}$ the entry of $\underline{G}$ in row $\sigma$ and column $\zeta$. We use analogous notation for matrices $\underline{R}$, $G$, and $R$. In general, if $\Gamma$ is any matrix, we will sometimes refer to the entry $\Gamma^{\sigma}_{\zeta}$ as the ``entry (indexed by) $(\sigma,\zeta)$.''

In general, if $\Gamma$ is a matrix with rows indexed by $\{0,1\}^a$ and columns indexed by $\{0,1\}^b$, and $S\subset\{0,1\}^a$ (resp. $S\subset\{0,1\}^b$), we will let $\Gamma^S$ (resp. $\Gamma_S$) denote the submatrix of $\Gamma$ consisting of rows (resp. columns) indexed by $S$. Where $\Gamma$ is clear from context, we will denote the row span of $\Gamma^S$ (resp. column span of $\Gamma_S$) by $\sp(S)$.

Lastly, a column/row is called \emph{odd} (resp. \emph{even}) if its index is odd (resp. even).

\subsection{Degenerate and Full Rank Signatures}

\begin{defn}
	A signature $G$ (generator or recognizer) is degenerate iff there exist vectors $\gamma_i$ ($1\le i\le n$) of dimension $k$ for which $G = \gamma_1\tensor\cdots\tensor\gamma_n$.
\end{defn}

\begin{lem}[Lemma 3.1, \cite{caifu}]
	A signature $G$ is degenerate iff $\rank(G(t))\le 1$ for $1\le t\le n$.\label{lem:caidegen}
\end{lem}

\begin{defn}
A signature $G$ is of full rank iff there exists some $1\le t\le n$ for which $\rank(G(t)) = k$.
\end{defn}

By Lemma~\ref{lem:relunderline}, if signature $G$ is of full rank, then for the corresponding standard signature $\underline{G}$, we have that $\rank(\underline{G}(t)) = k$ for some $t$. Over domain size 2, by Lemma~\ref{lem:caidegen}, all signatures not of full rank are degenerate, and holographic algorithms exclusively using such signatures are trivial because degenerate generators can by definition be decoupled into arity-1 generators. Over domain size $k\ge 3$ however, it is unknown to what extent holographic algorithms exclusively using signatures not of full rank trivialize. In \cite{caifu}, the collapse theorems over domain sizes 3 and 4 were proved under the assumption that at least one signature is of full rank, so we too make that assumption.

\subsection{Clusters}

One of the key results in our proof of the general collapse theorem is the existence within any matrix-form standard signature of a full-rank square submatrix whose entries have indices satisfying certain properties. In this section we make precise those properties.

\begin{defn}
	A set of $2^m$ distinct bitstrings $Z = \{x^1, ..., x^{2^m}\}\subset\{0,1\}^n$ is an \emph{$(m,n)$-cluster} if there exists $s\in\{0,1\}^n$ and positions $p_1, ..., p_m\in[n]$ such that for each $i\in[2^m]$, $x^i = s\oplus\left(\bigoplus_{j\in J}e_{p_j}\right)$ for some $J\subset\{p_1,...,p_m\}$. We write $Z$ as $s + \{e_{p_1},...,e_{p_m}\}$. Note that for a fixed cluster, $s$ is only unique up to the bits outside of positions $p_1,...,p_m$. If a cluster $Z'$ is a subset of another cluster $Z$, we say that $Z'$ is a \emph{subcluster} of $Z$.
\end{defn}

\begin{defn}
	In $\Gamma$, a $2^m\times 2^m$ submatrix $\Gamma'$ is a \emph{$m$-cluster submatrix} if there exist $(m,\ell)$- and $(m,(n-1)\ell)$-clusters $\Sigma$ and $Z$ such that $\Gamma' = (\Gamma^{\sigma}_{\zeta})^{\sigma\in\Sigma}_{\zeta\in Z}$ (here we omit the parameters $n$ and $\ell$ in the notation as they will be clear from context).
\end{defn}

\section{Matchgate Identities and Consequences}
\label{sec:mgis}

\subsection{Parity Condition and Matchgate Identities}

The most obvious property that standard signatures satisfy is the \emph{parity condition}: because a graph with an odd number of vertices has no perfect matchings, the indices of the nonzero entries in $\underline{G}$ have the same parity. It trivially follows that in $\underline{G}(t)$, columns $\underline{G}(t)_{\zeta}$ and $\underline{G}(t)_{\eta}$, if both nonzero, are linearly independent if $\zeta$ and $\eta$ are of opposite parities. The same holds for the rows of $\underline{G}(t)$.

In \cite{caigoren} it is shown algebraically that the parity condition in fact follows from the so-called \emph{matchgate identities} which we present below, quadratic relations which together form a necessary and sufficient condition for a vector to be the standard signature of some matchgate.

As in Cai and Fu's proof of the collapse theorem for domain size 4, we will make heavy use of the matchgate identities stated below. Wherever we invoke them, they will be for generator matchgates, so we focus on this case.

\begin{thm}[Theorem 2.1, \cite{caigoren}]
A $2^{\ell}\times 2^{(n-1)\ell}$ matrix $\Gamma$ is the $t$-th matrix form of the standard signature of some generator matchgate iff for all $\zeta,\eta\in\{0,1\}^{(n-1)\ell}$ and $\sigma,\tau\in\{0,1\}^{\ell}$, the following \emph{matchgate identity} holds. Let $\zeta\oplus\eta = e_{q_1}\oplus\cdots\oplus e_{q_{d'}}$ and $\sigma\oplus\tau = e_{p_1}\oplus\cdots\oplus e_{p_d}$, where $q_1<\cdots<q_{d'}$ and $p_1<\cdots<p_d$. Then \begin{equation}\sum^d_{i=1}(-1)^{i+1}\Gamma^{(\sigma\oplus e_{p_1}\oplus e_{p_i})}_{\zeta}\Gamma^{(\tau\oplus e_{p_1}\oplus e_{p_i})}_{\eta} = \sum^{d'}_{j=1}\pm\Gamma^{(\sigma\oplus e_{p_1})}_{(\zeta\oplus e_{q_j})}\Gamma^{(\tau\oplus e_{p_1})}_{(\eta\oplus e_{q_j})}.\label{eq:premgi}\end{equation} Here the $\pm$ signs depend on both $j$ and, if $q_j$ is after the $t$-th block, the parity of $d$. If $d$ is even, \begin{equation}\sum^d_{i=1}(-1)^{i+1}\Gamma^{(\sigma\oplus e_{p_1}\oplus e_{p_i})}_{\zeta}\Gamma^{(\tau\oplus e_{p_1}\oplus e_{p_i})}_{\eta} = \epsilon_{\zeta,\eta}\sum^{d'}_{j=1}(-1)^{j+1}\Gamma^{(\sigma\oplus e_{p_1})}_{(\zeta\oplus e_{q_j})}\Gamma^{(\tau\oplus e_{p_1})}_{(\eta\oplus e_{q_j})},\label{eq:mgis}\end{equation} where $\epsilon_{\zeta,\eta}\in\{\pm 1\}$ is positive (resp. negative) if the number of $q_j$ preceding the $t$-th block is odd (resp. even). \label{thm:mgis}
\end{thm}

\begin{remark}
	If $d+d'$ is odd, $\eqref{eq:premgi}$ is trivial by the parity condition.\label{remark:dumb}
\end{remark}

We will be making extensive use of the matchgate identities in this paper, but we will typically not care about the $\epsilon_{\zeta,\eta}$ sign on the right-hand side of \eqref{eq:mgis}. For this reason, it will be convenient to make the following definition.

\begin{defn}
	A $2^{\ell}\times 2^m$ matrix $M$ is a \emph{pseudo-signature} if for all $\sigma,\tau$ for which $\wt(\sigma\oplus\tau)$ is even, its entries satisfy the corresponding identity \eqref{eq:mgis} up to a factor of $\pm 1$ on the right-hand side.
\end{defn}

Standard signatures and cluster submatrices are all examples of pseudo-signatures.

\begin{obs}
	If $M$ is a pseudo-signature, then its transpose $M^T$ is a pseudo-signature.\label{obs:transpose}
\end{obs}

\subsection{Matchgate Identities and Determinants}

We now derive from the matchgate identities some basic linear algebraic properties of the columns of pseudo-signatures. By Observation~\ref{obs:transpose}, these also apply for the rows.

Firstly, we have the following immediate consequence of Theorem~\ref{thm:mgis}. We specifically consider the case where $\wt(\zeta\oplus\eta)$ is even and, by Remark~\ref{remark:dumb}, $\wt(\sigma\oplus\tau)$ is even. So write $\zeta\oplus\eta = e_{q_1}\oplus\cdots\oplus e_{q_{2d'}}$ and $\sigma\oplus\tau = e_{p_1}\oplus\cdots\oplus e_{p_{2d}}$.

Reverse the roles of $\zeta$ and $\eta$ in \eqref{eq:mgis}. Subtract the resulting equation from \eqref{eq:mgis} to find $$\sum^{2d}_{i=1}(-1)^{i+1}\left(\Gamma_{\eta}^{(\sigma\oplus e_{p_1}\oplus e_{p_i})}\Gamma_{\zeta}^{(\tau\oplus e_{p_1}\oplus e_{p_i})}-\Gamma_{\zeta}^{(\sigma\oplus e_{p_1}\oplus e_{p_i})}\Gamma_{\eta}^{(\tau\oplus e_{p_1}\oplus e_{p_i})}\right) = $$ $$\epsilon_{\zeta,\eta}\sum^{2d'}_{j=1}(-1)^{j+1}\left(\Gamma_{(\zeta\oplus e_{q_j})}^{(\sigma\oplus e_{p_1})}\Gamma_{(\eta\oplus e_{q_j})}^{(\tau\oplus e_{p_1})} - \Gamma_{(\eta\oplus e_{q_j})}^{(\sigma\oplus e_{p_1})}\Gamma_{(\zeta\oplus e_{q_j})}^{(\tau\oplus e_{p_1})}\right),$$ or equivalently, \begin{equation}\sum^{2d}_{i=1}(-1)^{i+1}\left|\begin{matrix}
\Gamma_{\zeta}^{(\sigma\oplus e_{p_1}\oplus e_{p_i})} & \Gamma_{\eta}^{(\sigma\oplus e_{p_1}\oplus e_{p_i})} \\
\Gamma_{\zeta}^{(\tau\oplus e_{p_1}\oplus e_{p_i})} & \Gamma_{\eta}^{(\tau\oplus e_{p_1}\oplus e_{p_i})}
\end{matrix}\right| = \epsilon_{\zeta,\eta}\sum^{2d'}_{j=1}(-1)^{j+1}\left|\begin{matrix}
\Gamma_{(\zeta\oplus e_{q_j})}^{(\sigma\oplus e_{p_1})} & \Gamma_{(\eta\oplus e_{q_j})}^{(\sigma\oplus e_{p_1})} \\
\Gamma_{(\zeta\oplus e_{q_j})}^{(\tau\oplus e_{p_1})} & \Gamma_{(\eta\oplus e_{q_j})}^{(\tau\oplus e_{p_1})}\end{matrix}\right|\label{eq:mgis2}\end{equation}

\begin{example}
\label{ex:mgi_basic}
	Suppose $d = d' = 1$. Then \eqref{eq:mgis2} becomes $$2\left|\begin{matrix}\Gamma^{\sigma}_{\zeta} & \Gamma^{\sigma}_{\eta} \\
	\Gamma^{\tau}_{\zeta} & \Gamma^{\tau}_{\eta}\end{matrix}\right| = 2\epsilon_{\zeta,\eta}\left|\begin{matrix}\Gamma^{\sigma\oplus e_{p_1}}_{\zeta\oplus e_{q_1}} & \Gamma^{\sigma\oplus e_{p_1}}_{\zeta\oplus e_{q_2}} \\
	\Gamma^{\sigma\oplus e_{p_2}}_{\zeta\oplus e_{q_1}} & \Gamma^{\sigma\oplus e_{p_2}}_{\zeta\oplus e_{q_2}}
	\end{matrix}\right|,$$ so in particular, the matrix on the left is singular iff the latter is.

	More generally, only suppose that $d = 1$. Then \eqref{eq:mgis2} becomes \begin{equation}2\left|\begin{matrix}\Gamma^{\sigma}_{\zeta} & \Gamma^{\sigma}_{\eta} \\
	\Gamma^{\tau}_{\zeta} & \Gamma^{\tau}_{\eta}\end{matrix}\right| = \epsilon_{\zeta,\eta}\sum^{2d'}_{j=1}(-1)^{j+1}\left|\begin{matrix}\Gamma^{(\sigma\oplus e_{p_1})}_{(\zeta\oplus e_{q_j})} & \Gamma^{(\sigma\oplus e_{p_1})}_{(\eta\oplus e_{q_j})} \\
	\Gamma^{(\sigma\oplus e_{p_2})}_{(\zeta\oplus e_{q_j})} & \Gamma^{(\sigma\oplus e_{p_2})}_{(\eta\oplus e_{q_j})}\end{matrix}\right|\label{eq:lindepeq}\end{equation} so in particular, the matrix on the left-hand side is singular if all $2d'$ matrices on the right-hand side are singular.\label{ex:basic2}
\end{example}

In other words, if $\ell = 2$ so that $\Gamma$ only has four rows, columns $\zeta$ and $\eta$ as defined above are linearly dependent if all pairs of neighboring columns are linearly dependent. We shall see in the next section (Corollary~\ref{cor:lindep}) that this is true even when $\Gamma$ has an arbitrary number of rows.

\subsection{Wedge Products of Columns}

Motivated by Example~\ref{ex:mgi_basic}, we'd like to study the set of all $2\times 2$ determinants $\left|\begin{matrix}
\Gamma^{\sigma}_{\zeta} & \Gamma^{\sigma}_{\eta} \\
\Gamma^{\tau}_{\zeta} & \Gamma^{\tau}_{\eta}
\end{matrix}\right|$ given two column vectors $\Gamma_{\zeta}$ and $\Gamma_{\eta}$ of the same parity. These are merely the coefficients of the wedge product $\Gamma_{\zeta}\wedge\Gamma_{\eta}$ under the standard basis $\{v_{\sigma}\wedge v_{\tau}\}_{\sigma,\tau\in\{0,1\}^{2^{\ell}}, \ \sigma<\tau}$ (where the relation $\sigma<\tau$ denotes lexicographic ordering) of $\bigwedge^2\C^{2^{\ell}}$, the second exterior power of $\C^{2^{\ell}}$. The matchgate identities imply the following consequence about the relationships among the wedge products $\Gamma_{\zeta}\wedge\Gamma_{\eta}$ as $\zeta$ and $\eta$ vary.

\begin{lem}
	If $\zeta_1,\eta_1,\zeta_2,\eta_2,...,\zeta_m,\eta_m$ are even indices for which \begin{equation}\sum^m_{\nu=1}a_{\nu}\left(\Gamma_{\zeta_{\nu}}\wedge\Gamma_{\eta_{\nu}}\right) = 0\label{eq:oddresult}\end{equation} for some $a_1,...,a_m\in\mathds{C}$, then \begin{equation}\sum^m_{\nu=1}\epsilon_{\zeta_{\nu},\eta_{\nu}}a_{\nu}\left(\sum^{2d_{\nu}}_{i=1}(-1)^{j+1}\Gamma_{\zeta_{\nu}\oplus e_{p^{\nu}_i}}\wedge\Gamma_{\eta_{\nu}\oplus e_{p^{\nu}_i}}\right) = 0,\label{eq:evenresult}\end{equation} where for each $\nu$, $\wt(\zeta_{\nu}\oplus\eta_{\nu})= 2d_{\nu}$ and $\zeta_{\nu}\oplus\eta_{\nu} = e_{p^{\nu}_1}\oplus\cdots\oplus e_{p^{\nu}_{2d_{\nu}}}$.
	\label{lem:oddimplieseven}
\end{lem}

\begin{proof}
	For convenience, we will denote $\epsilon_{\zeta_{\nu},\eta_{\nu}}$ by $\epsilon_{\nu}$. First, we rewrite \eqref{eq:evenresult} in terms of coordinates as \begin{equation}\sum^m_{\nu=1}\epsilon_{\nu}a_{\nu}\left(\sum^{2d_{\nu}}_{i=1}(-1)^{i+1}\sum_{\sigma<\tau}\left|\begin{matrix}
		\Gamma^{\sigma}_{\zeta_{\nu}\oplus e_{p^{\nu}_i}} & \Gamma^{\sigma}_{\eta_{\nu}\oplus e_{p^{\nu}_i}} \\
		\Gamma^{\tau}_{\zeta_{\nu}\oplus e_{p^{\nu}_i}} & \Gamma^{\tau}_{\eta_{\nu}\oplus e_{p^{\nu}_i}}
	\end{matrix}\right|(v_{\sigma}\wedge v_{\tau})\right) = 0,\label{eq:coords}\end{equation} where $\sigma<\tau$ denotes lexicographical ordering. Note that the determinants that appear in the left-hand side of \eqref{eq:coords} are zero when $\sigma$ and $\tau$ are of opposite parity. Moreover, depending on the parity of the signature $\Gamma$, either all such determinants are also zero for $\sigma$ and $\tau$ even, or they are all zero for $\sigma$ and $\tau$ odd.

	Rearranging the order of summations in \eqref{eq:coords}, the desired identity becomes \begin{equation}\sum_{\sigma<\tau}(v_{\sigma}\wedge v_{\tau})\cdot\left(\sum^m_{\nu=1}\epsilon_{\nu}a_{\nu}\sum^{2d_{\nu}}_{i=1}(-1)^{i+1}\left|\begin{matrix}
		\Gamma^{\sigma}_{\zeta_{\nu}\oplus e_{p^{\nu}_i}} & \Gamma^{\sigma}_{\eta_{\nu}\oplus e_{p^{\nu}_i}} \\
		\Gamma^{\tau}_{\zeta_{\nu}\oplus e_{p^{\nu}_i}} & \Gamma^{\tau}_{\eta_{\nu}\oplus e_{p^{\nu}_i}}
	\end{matrix}\right|\right) = 0.\label{eq:desired}\end{equation}For a fixed pair $\sigma<\tau$, let $\wt(\sigma\oplus\tau) = 2d'$ and $\sigma\oplus\tau = e_{q_1}\oplus\cdots\oplus e_{q_{2d'}}$. If we apply \eqref{eq:mgis2} and rearrange the order of summation once more, the coefficient of $v_{\sigma}\wedge v_{\tau}$ above becomes
	$$\sum^m_{\nu=1}a_{\nu}\cdot\sum^{2d'}_{j=1}(-1)^{j+1}\left|\begin{matrix}
		\Gamma^{\sigma\oplus e_{q_j}}_{\zeta_{\nu}} & \Gamma^{\sigma\oplus e_{q_j}}_{\eta_{\nu}} \\
		\Gamma^{\tau\oplus e_{q_j}}_{\zeta_{\nu}} & \Gamma^{\tau\oplus e_{q_j}}_{\eta_{\nu}}
	\end{matrix}\right| = \sum^{2d'}_{j=1}(-1)^{j+1}\sum^m_{\nu=1}a_{\nu}\left|\begin{matrix}
		\Gamma^{\sigma\oplus e_{q_j}}_{\zeta_{\nu}} & \Gamma^{\sigma\oplus e_{q_j}}_{\eta_{\nu}} \\
		\Gamma^{\tau\oplus e_{q_j}}_{\zeta_{\nu}} & \Gamma^{\tau\oplus e_{q_j}}_{\eta_{\nu}}
	\end{matrix}\right|.$$ But note that if we expand \eqref{eq:oddresult} in terms of coordinates, the term $$a_{\nu}\left|\begin{matrix}
			\Gamma^{\sigma\oplus e_{q_j}}_{\zeta_{\nu}} & \Gamma^{\sigma\oplus e_{q_j}}_{\eta_{\nu}} \\
			\Gamma^{\tau\oplus e_{q_j}}_{\zeta_{\nu}} & \Gamma^{\tau\oplus e_{q_j}}_{\eta_{\nu}}
		\end{matrix}\right|$$ is precisely the coefficient of $v_{\sigma\oplus e_{q_j}}\wedge v_{\tau\oplus e_{q_j}}$ in the expansion of \eqref{eq:oddresult} in terms of coordinates and hence zero by assumption, so \eqref{eq:desired} holds as desired.
\end{proof}

\begin{cor}
	Let $\zeta,\eta\in\{0,1\}^{(n-1)\ell}$ be such that $\zeta\oplus\eta = \bigoplus^{2d}_{j=1}e_{p_i}$. If column $\Gamma_{(\zeta\oplus e_{p_i})}$ is linearly dependent with column $\Gamma_{(\eta\oplus e_{p_i})}$ for $1\le i\le 2d$, then column $\Gamma_{\zeta}$ is linearly dependent with column $\Gamma_{\eta}$.\label{cor:lindep}
\end{cor}

\begin{cor}
	Let $\zeta,\eta\in\{0,1\}^{(n-1)\ell}$ be such that $\zeta\oplus\eta = \bigoplus^{2d}_{i=1}e_{p_i}$. If there exists $i\in[2d']$ such that column $\Gamma_{\zeta\oplus e_{p_i}}$ is linearly dependent with column $\Gamma_{\eta\oplus e_{p_i}}$ for $i = 1,...,\hat{j},...,2d'$, where $\hat{j}$ denotes omission of index $j$, and if $\Gamma_{\zeta}$ is also linearly dependent with column $\Gamma_{\eta}$, then $\Gamma_{\zeta\oplus e_{p_j}}$ and $\Gamma_{\eta\oplus e_{p_j}}$ are linearly dependent.\label{cor:lindepprime}
\end{cor}

Lemma~\ref{lem:oddimplieseven} says that any linear relation among wedges of even columns yields a linear relation among wedges of odd columns, and vice versa.

Lastly, we need the following elementary result in multilinear algebra.

\begin{lem}
	If $v_1,...,v_n$ are linearly independent in vector space $V$, then the set of all $v_i\wedge v_j$ for $i<j$ are linearly independent in $\bigwedge^2 V$.\label{lem:wedge}
\end{lem}

Combining this with Lemma~\ref{lem:oddimplieseven} yields the following key ingredient to the analysis in Section~\ref{sec:rigidclust}.

\begin{lem}
	Suppose $\zeta_0,\eta\in\{0,1\}^K$ such that $\zeta_0\neq\eta$, and the indices in $T = \{\zeta_1,...,\zeta_m\}\subset\{0,1\}^K$ are distinct and have the same parity. Suppose further that $\zeta_0\neq\zeta_1,...,\zeta_m$. Let $\zeta_i\oplus\eta = e_{p^i_1}\oplus\cdots\oplus e_{p^i_{d_i}}$ for $0\le i\le m$, where $d_i:=\wt(\zeta_i\oplus\eta)$. Define $$S:=\bigcup_{0\le i\le m, 1\le j\le d_i}\{\eta\oplus e_{p^i_j},\zeta_i\oplus e_{p^i_j}\}\subset\{0,1\}^{K}$$ not in the sense of multisets, that is, we throw out duplicates so that the strings in $S$ are all distinct.

	Suppose the columns indexed by $S$ are linearly independent. Then \begin{equation}\Gamma_{\zeta_0}\not\in\sp(\Gamma_{\zeta_1},...,\Gamma_{\zeta_m})\label{eq:notinspan}\end{equation}

	If $\wt(\eta\oplus\zeta_0)\ge 4$ and $j^*\in[d_0]$, then \eqref{eq:notinspan} holds even if only the columns indexed by $S':=S\backslash\{\zeta_0\oplus e_{p^0_{j^*}}\}$ are linearly independent.\label{lem:key}
\end{lem}

\begin{proof}
	We first prove the claim without the assumption that $\wt(\eta\oplus\zeta_0)\ge 4$. Suppose to the contrary that $\Gamma_{\zeta_0} = \sum^m_{i=1}a_i\Gamma_{\zeta_i}$ so that $\Gamma_{\zeta_0}\wedge\Gamma_{\eta} - \sum^m_{i=1}a_i\Gamma_{\zeta_i}\wedge\Gamma_{\eta} = 0$. By Lemma~\ref{lem:oddimplieseven}, this linear relation implies the following linear relation among wedges of columns of the other parity: \begin{equation}
		\epsilon_{\zeta_0,\eta}\left(\sum^{d_0}_{j=1}(-1)^{j+1}\Gamma_{\zeta_0\oplus e_{p^0_j}}\wedge\Gamma_{\eta\oplus e_{p^0_j}}\right) - \sum^m_{i=1}\epsilon_{\zeta_i,\eta}a_i\left(\sum^{d_i}_{j=1}(-1)^{j+1}\Gamma_{\zeta_i\oplus e_{p^i_j}}\wedge\Gamma_{\eta\oplus e_{p^i_j}}\right) = 0
		\label{eq:otherparity}
	\end{equation}

	We claim this is a nontrivial linear relation contradicting the linear independence of the columns indexed by $S$. For each of the $m+1$ sums indexed by $1\le j\le d_i$ appearing in \eqref{eq:otherparity}, if $d_i = 2$, rewrite $\sum^{d_i}_{j=1}(-1)^{j+1}\Gamma_{\zeta_i\oplus e_{p^i_j}}\wedge\Gamma_{\eta\oplus e_{p^i_j}}$ as $2\Gamma_{\zeta_i\oplus e_{p^i_1}}\wedge\Gamma_{\eta\oplus e_{p^i_1}}$.

	After this consolidation, note that the wedge products in \eqref{eq:otherparity} are now all distinct. Certainly for any $j,j'\in[d_i]$ where $d_i\ge 4$, $\Gamma_{\zeta_i\oplus e_{p^i_j}}\wedge\Gamma_{\eta\oplus e_{p^i_j}}$ and $\Gamma_{\zeta_i\oplus e_{p^i_{j'}}}\wedge\Gamma_{\eta\oplus e_{p^i_{j'}}}$ are linearly independent. For $i,i'$ such that $d_i = 2$ and $d_{i'}>2$, $2\Gamma_{\zeta_i\oplus e_{p^i_1}}\wedge\Gamma_{\eta\oplus e_{p^i_1}}$ and any $\Gamma_{\zeta_{i'}\oplus e_{p^{i'}_j}}\wedge\Gamma_{\eta\oplus e_{p^{i'}_j}}$ are linearly independent as $\zeta_i\oplus\eta\neq\zeta_{i'}\oplus\eta$, contradicting the assumption that $\{\zeta_1,...,\zeta_m\}$ are distinct. Similarly, for $i,i'$ such that $d_i>2$ and $d_{i'}>2$, any $\Gamma_{\zeta_i\oplus e_{p^i_j}}\wedge\Gamma_{\eta\oplus e_{p^i_j}}$ and any $\Gamma_{\zeta_{i'}\oplus e_{p^{i'}_{j'}}}\wedge\Gamma_{\eta\oplus e_{p^{i'}_{j'}}}$ are linearly independent as $\zeta_i\oplus\eta\neq \zeta_{i'}\oplus\eta$.

	We conclude that \eqref{eq:otherparity}, after consolidating sums for which $d_i = 2$, consists of a nonzero number of linearly independent wedge products of columns indexed by $S$, so \eqref{eq:otherparity} is indeed a nontrivial linear relation among the wedge products $\Gamma_s\wedge\Gamma_{s'}$ for $s, s'\in S$. But all columns indexed by $S$ are linearly independent by assumption, so this linear relation contradicts Lemma~\ref{lem:wedge} and the linear independence of columns indexed by $S$.

	For the second part of Lemma~\ref{lem:key}, we claim that \eqref{eq:otherparity} is still a nontrivial relation. Pick any $k\neq j^*$ inside $[d_0]$. Because $d_0\ge 4$, \begin{equation}\zeta_0\oplus e_{p^0_k}, \eta\oplus e_{p^0_k}\neq \eta\oplus e_{p^0_{j^*}}.\label{eq:neq}\end{equation}

	We have already taken care of the case where $\Gamma_{\zeta_0 + e_{p^0_{j^*}}}\not\in\sp(S')$ above, so suppose instead that \begin{equation}\Gamma_{\zeta_0 + e_{p^0_{j^*}}} = \sum_{s\in S'}b_s\Gamma_s\label{eq:substitute}.\end{equation} If we consolidate sums for which $d_i = 2$ in \eqref{eq:otherparity} as above and substitute \eqref{eq:substitute} into the resulting equation, the wedge products that \eqref{eq:otherparity} now contains also include ones of the form $\Gamma_s\wedge\Gamma_{\eta+e_{p^0_{j^*}}}$, which cannot be linearly dependent with $\Gamma_{\zeta_0\oplus e_{p^0_k}}\wedge\Gamma_{\eta\oplus e_{p^0_k}}$ by \eqref{eq:neq}.

	Every other wedge product in \eqref{eq:otherparity} is of the form $\Gamma_{\zeta_i\oplus e_{p^i_j}}\wedge\Gamma_{\eta\oplus e_{p^i_j}}$ and also cannot be linearly dependent with $\pm\Gamma_{\zeta_0\oplus e_{p^0_k}}\wedge\Gamma_{\eta\oplus e_{p^0_k}}$ or else, as before, we'd find that $\zeta_i\oplus\eta = \zeta_k\oplus\eta$, contradicting the assumption that $\{\zeta_1,...,\zeta_m\}$ are distinct.

	It follows that if we rewrite the left-hand side of \eqref{eq:otherparity} in the form $\sum_{s,s'\in S'}b_{s,s'}\Gamma_{s}\wedge\Gamma_{s'}$ (uniquely because by Lemma~\ref{lem:wedge} the $\Gamma_s\wedge\Gamma_{s'}$ are linearly independent), $b_{\zeta_0\oplus e_{p^0_k},\eta\oplus e_{p^0_k}}\neq 0$. So \eqref{eq:otherparity} is still a nontrivial linear relation, contradicting Lemma~\ref{lem:wedge} and the linear independence of columns indexed by $S$.
\end{proof}

\section{Rigidity and Cluster Existence}
\label{sec:rigidclust}

In this section, we will prove the rigidity theorem and the cluster existence theorem, informally stated as Theorems~\ref{thm:rigid_informal} and \ref{thm:sub_informal}. Now that we have introduced the appropriate terminology, we first state both precisely.

\begin{thm}[Rank Rigidity - formal]
	If $\Gamma$ is a $2^{\ell}\times 2^m$ pseudo-signature, then $\rank(\Gamma)$ is a power of 2. Equivalently, for all $\kappa\ge 1$, \begin{equation}\rank(\Gamma)\ge 2^{\kappa}+1 \Rightarrow\rank(\Gamma)\ge 2^{\kappa+1}\label{eq:implication1}\end{equation}
\label{thm:rigid}
\end{thm}

\begin{thm}[Cluster Existence - formal]
	If $\Gamma = \underline{G}(t)$ for some matchgate $G$ of arity $n\ell$, and $\rank(\Gamma)\ge k$, then there is a $\lceil\log_2 k\rceil$-cluster submatrix of full rank.\label{thm:corsub}
\end{thm}

To prove Theorem~\ref{thm:corsub}, we claim it is enough to show the following:

\begin{thm}
	If $\Gamma$ is a $2^{\ell}\times 2^{m}$ pseudo-signature such that $\rank(\Gamma)\ge k$, then there exists a $(\lceil\log_2 k\rceil, m)$-cluster $Z$ for which $\Gamma_Z$ is of full rank.

	Equivalently, for all $\kappa\ge 1$, \begin{equation}
	\rank(\Gamma)\ge 2^{\kappa}+1 \Rightarrow \exists \ (\kappa+1, m)\text{-cluster} \ Z \ \text{such that} \ \rank(\Gamma_Z) = 2^{\kappa+1}\label{eq:implication2}\end{equation}
\label{thm:sub}
\end{thm}

\begin{proof}[Proof of Theorem~\ref{thm:corsub}]
	Apply Theorem~\ref{thm:sub} to $\Gamma$ to obtain $\Gamma_Z$. By Observation~\ref{obs:transpose}, $\Gamma_Z^T$ is also a pseudo-signature. Apply Theorem~\ref{thm:sub} to $\Gamma^T$ to get the desired cluster submatrix.
\end{proof}

Note that for fixed $\kappa$, \eqref{eq:implication2} implies \eqref{eq:implication1}. We will jointly prove Theorem~\ref{thm:rigid} and Theorem~\ref{thm:sub} by induction on $k$. Cai and Fu have already shown both for $k\le 4$; we take these results as our base case. We complete the following two inductive steps.

\begin{indstep}
If implication~\eqref{eq:implication2} holds for $1\le\kappa\le K-1$, then implication \eqref{eq:implication1} holds for $1\le \kappa\le K$. \label{lem:indstep1}
\end{indstep}

\begin{indstep}
If implication~\eqref{eq:implication1} holds for $1\le \kappa\le K$ and implication~\eqref{eq:implication2} holds for $\kappa\le K-1$, then implication \eqref{eq:implication2} also holds for $\kappa = K$.\label{lem:indstep2}
\end{indstep}

Note that once we have proven the rigidity and cluster existence theorems, we additionally obtain the following.

\begin{cor}
	If $G$ is a full-rank matchgate signature on domain size $k$, it is only realizable over bases $M$ of rank at most $2^{\lfloor\log_2 k\rfloor}$.\label{cor:red}
\end{cor}

\begin{proof}
	If $k$ is a power of two, the claim is trivial. If not, assume $G$ is a generator (standard signatures of recognizer are also pseudo-signatures, so the argument in that case is analogous). If $\rank(G(t)) = k$ and $\rank(M)\ge 2^{\lfloor\log_2 k\rfloor} + 1$, then by Lemma~\ref{lem:relunderline2}, $2^{\lfloor\log_2 k\rfloor} + 1\le \rank(\underline{G}(t))\le k$. But Theorem~\ref{thm:rigid} would then imply $\rank(\underline{G}(t))\ge 2^{\lfloor\log_2 k\rfloor+1}>k$, a contradiction.
\end{proof}

\subsection{Rank Rigidity Theorem}
\label{subsec:rigid}

In this subsection, we complete the former inductive step, and in the next, we complete the latter.

Before we prove Inductive Step~\ref{lem:indstep1} in its entirety, we take care of the case where $m = K + 1$. While this might appear to be the simplest case because $m$ is minimal, it will turn out that cases where $m$ is greater will reduce to this case. For this reason, the wedge product machinery introduced in Section~\ref{sec:mgis} is used exclusively in the proof of this case.

Once we take care of this case, we will essentially show that if standard signature $\Gamma$ is any wider, i.e. if $m>K+1$, then if $Z$ is a cluster of size greater than $2^K$ indexing columns of rank at least $2^K + 1$, we can always find a proper subcluster also indexing columns of rank at least $2^K+1$, or else the matchgate identities would erroneously imply that certain columns which are known to be linearly independent are linearly dependent.

We begin with the case of $m = K+1$.

\begin{thm}
	If $\Gamma$ is a $2^{\ell}\times 2^{K+1}$ pseudo-signature such that $\rank(\Gamma)\ge 2^{K} + 1$, then $\rank(\Gamma)= 2^{K+1}$.\label{thm:mini_rigid}
\end{thm}

\begin{proof}
	Because we assume implication~\eqref{eq:implication2} holds for $\kappa = K-1$, $\Gamma$ contains a $(K,m)$-cluster $Z = s\oplus\{e_1,...,\hat{e_j},...,e_{K+1}\}$ of linearly independent columns; denote the even indices of $Z$ by $Z_0$ and the odd ones by $Z_1$. Because $\rank(\Gamma) > 2^{K}$, there exists $t\not\in Z$ for which $\Gamma_t\not\in\sp(Z)$. Denote the parity of $t$ by $b\in\{0,1\}$, and denote by $\overbar{b}$ the opposite parity.

	Select any $t' = t\oplus e_{i^*}$ for ${i^*}\neq j$ and apply Lemma~\ref{lem:key} to $\zeta_0 = t'$, $\eta = t\oplus e_j$, $T = Z_{\overbar{b}}$ to conclude that $\Gamma_{t'}\not\in\sp({Z_{\overbar{b}}})$.

	Let $S_{d_1,d_2}$ denote the set of column indices $\zeta\not\in Z$ for which $\wt(\zeta\oplus t) = d_1$ and $\wt(\zeta\oplus t') = d_2$. Note that because $\wt(t\oplus t') = 1$, $S_{d_1,d_2}$ is empty if $|d_1-d_2|\neq 1$. 

	We will show by induction on $D$ that the columns indexed by $\left(\bigcup^{D}_{d=0}S_{d,d+1}\cup S_{d+1,d}\right)\cup Z$ are linearly independent for all $0\le D\le K$. The definition of $t$ and the argument above for $t'$ give the base case of $D = 0$.

	For the inductive step, for each $d$ let $d_0$ and $d_1$, respectively denote the even and odd value in $\{d,d+1\}$. As columns indexed by $S_{d_0,d_1}$ and columns indexed by $S_{d_1,d_0}$ have opposite parity, it suffices to show that the columns indexed by $T^0_D := \left(\bigcup^{D}_{d=0}S_{d_0,d_1}\right)\cup Z_b$ are linearly independent, and that the columns indexed by $T^1_D := \left(\bigcup^D_{d=0}S_{d_1,d_0}\right)\cup Z_{\overbar{b}}$ are linearly independent.

	Within this inductive step, we will further induct on the elements within $S_{D_0,D_1}$ and $S_{D_1,D_0}$. Specifically, suppose we have already proven that for some subset $S'_{D_0,D_1}\subset S_{D_0,D_1}$, all columns indexed by $T^0_{D-1}\cup S'_{D_0,D_1}$ are linearly independent, and that for $S'_{D_1,D_0} := \{u\oplus e_{i^*}: \ u\in S'_{D_0,D_1}\}\subset S_{D_1,D_0}$, all columns indexed by $T^1_{D-1}\cup S'_{D_1,D_0}$ are linearly independent. Select any $u\not\in S'_{D_0,D_1}$ and apply Lemma~\ref{lem:key} to $\zeta_0 = u$, $\eta = t'\oplus e_j$, $T = T^0_{D-1}\cup S'_{D_0,D_1}$ to see that $\Gamma_u\not\in\sp(T^0_{D-1}\cup S'_{D_0,D_1})$. Note that when $\wt(\zeta_0\oplus\eta)\ge 4$, we do not yet know that $\Gamma_{\zeta_0\oplus e_{i^*}} = \Gamma_{u\oplus e_{i^*}}$ lies outside $\sp(T^1_{D-1}\cup S'_{D_1,D_0})$, that is, we do not know whether all the columns indexed by the set $S$ defined in Lemma~\ref{lem:key} are linearly independent, but the second part of Lemma~\ref{lem:key} says that we may still conclude that $\Gamma_u\not\in\sp(T^0_{D-1}\cup S'_{D_0,D_1})$ because the columns indexed by $S\backslash\{u\oplus e_{i^*}\}$ are linearly independent.

	Lastly, apply Lemma~\ref{lem:key} to $\zeta_0 = u\oplus e_{i^*}$, $\eta = t\oplus e_j$, $T = T^1_{D-1}\cup S'_{D_1,D_0}$ to see that $\Gamma_{u\oplus e_{i^*}}\not\in\sp({T^1_{D-1}\cup S'_{D_1,D_0}})$. Note that here we only need to invoke the first part of Lemma~\ref{lem:key} we already know that $\Gamma_{\zeta_0\oplus e^*} = \Gamma_u$ lies outside $\sp(T^0_{D-1}\cup S'_{D_0,D_1})$.
\end{proof}

We are now ready to complete Inductive Step 1.

\begin{proof}[Completion of Inductive Step~\ref{lem:indstep1}]
	As we remarked earlier, implication \eqref{eq:implication2} for a fixed value of $\kappa$ implies implication \eqref{eq:implication1} for that value of $\kappa$, so we just need to show that implication \eqref{eq:implication1} also holds for $\kappa = K$.

	Suppose to the contrary that there exists pseudo-signature $\Gamma$ of rank $k$ such that $2^K+1\le k < 2^{K+1}$. Without loss of generality, we may assume that for all clusters $Z\subsetneq\{0,1\}^m$, $\rank(\Gamma_Z) \le 2^K$; otherwise, replace $\Gamma$ by $\Gamma_Z$ for some small enough cluster $Z$ such that $\rank(\Gamma_Z)\ge 2^K+1$ and $\rank(\Gamma_{Z'})\le 2^K$ for all subclusters $Z'\subsetneq Z$. Furthermore, by Theorem~\ref{thm:mini_rigid}, we may assume $m > K+1$.

	\begin{lem}
		If $Z = s+\{e_{p_1},...,e_{p_{K}}\}$ is a $(K, m)$-cluster of linearly independent columns in $\Gamma$, then any column $\Gamma_t$ for which $t_i = s_i$ for some $i\neq p_1,...,p_{K}$ lies in the span of the columns indexed by $Z$.
		\label{lem:inspan}
	\end{lem}

	\begin{proof}
		If to the contrary there existed such a $\Gamma_t$ not lying in the span of $Z$ so that $t_i = s_i$ for some $i\neq p_1,...,p_{K}$, then if $Z'$ is the $(m-1,m)$-cluster of column indices $\zeta$ for which $\zeta_i = s_i$, $\Gamma_{Z'}$ contains $t$ and all of $Z$ and thus has rank at least $2^{K} + 1$, contradicting our assumption on the ranks of the proper subclusters of $\Gamma$.
	\end{proof}

	By the inductive hypothesis, $\Gamma$ contains a $(K,m)$-cluster $Z$ of linearly independent columns $s+\{e_{p_1},...,e_{p_{K}}\}$. As $s$ is only uniquely defined modulo the bits in positions $p_1,...,p_K$, we will leave those bits of $s$ unspecified for now.

	Because $k>2^{K}$, there exists $t\not\in Z$ for which all columns indexed by $Z\cup\{t\}$ are linearly independent. Moreover, by Lemma~\ref{lem:inspan}, $t_j = \overbar{s_j}$ for all $j\not\in\{p_1,...,p_K\}$. Let $Z'$ denote the cluster $t + \{e_{p_1},...,e_{p_K}\}$. Set $s_i = t_i$ for $i\in\{p_2,...,p_K\}$; we will set $s_{p_1}$ to be 0 or 1 depending on the parity of the number $m-K$ of bits outside of positions $p_1,...,p_K$.

	\begin{case}
		$m-K$ is even.
	\end{case}

	Set $s_{p_1} = t_{p_1}$ so that $s$ and $t$ have the same parity.

	\begin{claim}
		If $u\not\in Z'$ and $u_i = s_i$ for each $i\in\{p_1,...,p_{K}\}$, then $\Gamma_{u}$ and $\Gamma_{s}$ are linearly dependent.\label{claim:001}
	\end{claim}

	\begin{proof}
		For each $i\in\{p_1,...,p_{K}\}$ and $j\not\in\{p_1,...,p_{K}\}$, let $T_i$ denote the cluster of all column indices $u$ for which $u_i = s_i$, and let $T^j_i$ denote the cluster of all column indices $u$ for which $u_i = s_i$ and $u_j = s_j$. Let $Z_i = Z\cap T_i$; obviously $Z_i\subset T^j_i\subset T_i$.


		Because $T^j_i$ is a cluster properly contained in $\{0,1\}^m$, we inductively know that $\rank(\Gamma_{T^j_i})\le 2^K$. And because $Z_i\subset T^j_i$, $\rank(\Gamma_{T^j_i})\ge 2^{K-1}$. But if $\rank(\Gamma_{T^j_i})\ge 2^{K-1}+1$, then by inductive hypothesis \eqref{eq:implication2} applied to $\Gamma_{T^j_i}$ for $\kappa = K-1$, $\rank(\Gamma_{T^j_i})\ge 2^K$. In other words, $\rank(\Gamma_{T^j_i})$ is either $2^{K-1}$ or $2^K$. 

		We will show that the latter is impossible. Suppose to the contrary that $\rank(\Gamma_{T^j_i}) = 2^{K}$.

		Then because $T^j_i\subset T_i$ and $\rank(\Gamma_{T_i}) = 2^K = \rank(\Gamma_{T^j_i})$, it follows that $\sp(T_i)=\sp(T^j_i)$. For any $u\in T^j_i$, $\Gamma_u\in\sp(Z)$ by Lemma~\ref{lem:inspan}, so $$\sp(Z)\supset\sp(T^j_i) = \sp(T_i).$$ But $T_i$ contains $t$, and by definition $\Gamma_t\not\in\sp(Z)$, a contradiction.

		We conclude that $\rank(\Gamma_{T^j_i}) = 2^{K-1}$. Then because $Z_i\subset T^j_i$ and $\rank(\Gamma_{Z_i}) = 2^{K-1} = \rank(\Gamma_{T^j_i})$, it follows that $\sp(Z_i) = \sp(T^j_i)$.

		In particular, all columns indexed by $\bigcap^K_{k=1}T^j_{p_k}$ lie in $\bigcap^K_{k=1}\sp(Z_{p_k}) = \sp(\{s\})$. Our choice of $j$ was arbitrary, so we get the desired claim.
	\end{proof}

	From the above claim and the fact that we're assuming $m > K+1$, we conclude that $\Gamma_{s\oplus e_j}$ for any $j\not\in\{p_1,...,p_{K}\}$ lies in the span of $\Gamma_s$. But $s$ and $s\oplus e_j$ are of opposite parity, so by the parity condition, $\Gamma_{s\oplus e_j} = 0$ for all such $j$. Applying Corollary~\ref{cor:lindep} to $s$ and $t$, it follows that $\Gamma_s$ and $\Gamma_{t}$ are linearly dependent, a contradiction.

	\begin{case}
		$m-K$ is odd.
	\end{case}

	Set $s_{p_1} = \overbar{t_{p_1}}$ so that $s$ and $t$ have the same parity.

	\begin{claim}
		For any $u\in\{0,1\}^m$, if $u\not\in Z'$ and $u_i = s_i$ for all $i\in\{p_2,...,p_{K}\}$, then:

		\begin{enumerate}
		 	\item If $u$ and $s$ have the same parity, then $\Gamma_{u}$ and $\Gamma_{s}$ are linearly dependent.

		 	\item If $u$ and $s$ have the opposite parity, then $\Gamma_u$ and $\Gamma_{s\oplus e_{p_1}}$ are linearly dependent.
		 \end{enumerate} \label{claim:002}
	\end{claim}

	\begin{proof}
		The proof is the same as that of Claim~\ref{claim:001}, the only subtlety being that $s$ and $t$ now only necessarily agree on bits $p_2,...,p_{K}$. By the argument there, all columns indexed by $T^j_i$ lie in $\sp(Z_i)$ for $i = p_2,...,p_{K}$. In particular, for all $j\not\in\{p_1,...,p_K\}$, all columns indexed by $\bigcap^{K}_{k=2}S^j_{p_k}$ lie in $\bigcap^{K}_{k=2}\sp(Z_{p_k}) = \sp(\{s,s\oplus e_{p_1}\})$.

		So given $u\not\in Z'$, write $\Gamma_u = \alpha\Gamma_s + \beta\Gamma_{s\oplus e_{p_1}}$. If $u$ and $s$ have the same parity, $\beta = 0$ by the parity condition, so $\Gamma_u\in\sp(\{s\})$. If $u$ and $s$ have the opposite parity, $\alpha = 0$ by the parity condition, so $\Gamma_u\in\sp(\{s\oplus e_{p_1}\})$.
	\end{proof}

	Pick any $j\not\in\{p_1,...,p_K\}$ and define $s^* = s\oplus e_{j}$ and $t^* = t\oplus e_j$. $s^*$ and $t^*$ both satisfy the hypotheses of Claim~\ref{claim:002} and have parity opposite to that of $s$, so by the latter case of Claim~\ref{claim:002}, they are both linearly dependent with $\Gamma_{s\oplus e_{p_1}}$. But $\Gamma_{s\oplus e_{p_1}}\neq 0$ because $s\oplus e_{p_1}\in Z$ and the columns indexed by $Z$ are linearly independent, so $\Gamma_{s^*}$ and $\Gamma_{t^*}$ are linearly dependent with each other.

	To show $\Gamma_s$ and $\Gamma_t$ are linearly dependent, we wish to apply Corollary~\ref{cor:lindepprime} to $s^*, t^*$, noting that $s^*\oplus t^* = e_j\oplus\sum_{j'\not\in\{j,p_2,...,p_k\}}e_j$.

	For any $j'\not\in\{j,p_2,...,p_K\}$, note that $s^*\oplus e_{j'}$ and $t^*\oplus e_{j'}$ both satisfy the hypotheses of Claim~\ref{claim:002} and have the same parity as $s$, so by the former case of Claim~\ref{claim:002}, they are both linearly dependent with $\Gamma_s$. But $\Gamma_{s}\neq 0$ because $s\in Z$ and the columns indexed by $Z$ are linearly independent, so $\Gamma_{s^*\oplus e_{j'}}$ and $\Gamma_{t^*\oplus e_{j'}}$ are linearly dependent with each other.

	Applying Corollary~\ref{cor:lindepprime} to $s^*$ and $t^*$, it follows that $\Gamma_s$ and $\Gamma_t$ are linearly dependent, a contradiction.
\end{proof}

\subsection{Existence of Cluster Submatrix}
\label{subsubsec:collapse-exist}

\begin{proof}[Completion of Inductive Step~\ref{lem:indstep2}]
	As in inductive step 1, we may assume without loss of generality that for all clusters $Z\subsetneq\{0,1\}^m$, $\rank(\Gamma_Z)\le 2^K$. If $m = K+1$, then by the inductive hypothesis that \eqref{eq:implication1} holds for $\kappa = K$, we're done. So suppose $m > K+1$.

	By the second part of the inductive hypothesis, implication~\eqref{eq:implication2} holds for $1\le\kappa\le K-1$, so $\Gamma$ contains a $(K,m)$-cluster $Z$ of linearly independent columns $s + \{e_{p_1},...,e_{p_{K}}\}$.

	As in inductive step 1, we can apply Lemma~\ref{lem:inspan} to $Z$ to see that all columns outside the span of the columns indexed by $Z$ must be indexed by $Z' = t + \{e_{p_1},...,e_{p_{K}}\}$, where $t = \left(s\oplus\bigoplus_{i\neq p_1,...,p_{K}}e_i\right)$. But $|Z'| = |Z| = 2^{K}$, and $\rank(\Gamma)\ge 2^{K+1}$ by implication~\eqref{eq:implication1} for $\kappa = K$, so the columns indexed by $Z\cup Z'$ are linearly independent. Because $m > K+1$, there exist columns not indexed by either $Z$ or $Z'$, and by Lemma~\ref{lem:inspan} applied once to $Z$ and once to $Z'$, these columns are in both $\sp(Z)$ and $\sp(Z')$ and thus must be zero.

	If $s$ and $t$ are of the same parity, apply Corollary~\ref{cor:lindep} to $s$ and $t$ to find that $\Gamma_s$ and $\Gamma_{t}$ are linearly dependent, a contradiction.

	If $s$ and $t$ are of opposite parity, apply Corollary~\ref{cor:lindepprime} to $s\oplus e_j$ and $t\oplus e_j\oplus e_{p_1}$ for any $j\not\in\{p_1,...,p_{K}\}$ to find that $\Gamma_s$ and $\Gamma_{t\oplus e_{p_1}}$ are linearly dependent, a contradiction.
\end{proof}

\section{Group Property of Standard Signatures}
\label{sec:collapse-group}

We will now prove the following generalization of the group property result over domain size 4 due to Cai and Fu (Theorem 5.5, \cite{caifu}):

\begin{thm}
	If $G$ is a generator matchgate of arity $Kn$ with standard signature $\underline{G}$, and $\rank(\underline{G}(t)) = 2^K$ for some $t$, then there exists a recognizer matchgate of arity $Kn$ such that $\underline{G}(t)\underline{R}(t) = I_{2^K}$.\label{thm:main_step}
\end{thm}

Roughly, we invoke Theorem~\ref{thm:sub} to obtain a full-rank $K$-cluster submatrix $G'$ of $\underline{G}(t)$ with column indices belonging to cluster $\zeta + \{e_{p_1},...,e_{p_K}\}$. Assume without loss of generality that $\zeta_{p_i} = 0$ for all $i\in[K]$. We will show that the matrix obtained by replacing $G'$ in $\underline{G}(t)$ with $(G')^{-1}$ and the remaining entries with zeroes is the standard signature of some arity-$Kn$ recognizer.

We first fix some notation. Denote $\underline{G}(t)$ by $\Gamma$. Suppose that nodes $p_1<\cdots <p_m \in [Kn]$ belong to blocks before the $t$-th, and nodes $p_{m+1}<\cdots<p_K\in [Kn]$ belong to blocks after the $t$-th. For expository purposes, we wish to use a particular permutation $(q_1,...,q_K)$ of $(p_1,...,p_K)$, so for $i\le m$, let $q_i = p_{m-i+1}$, and for $i>m$, let $q_i = p_{K+m-i+1}$ (see Figure~\ref{fig:toy}). If the column indices of $\Gamma$ are of the form $i_1\cdots i_{K(n-1)}$, those of $G'$ are of the form $i_{p_1}\cdots i_{p_K}$.

In \cite{groupprop}, Li and Xia gave a constructive proof that in the character theory of matchgates, the $2^{K}\times 2^{K}$ \emph{character} matrices of invertible $K$-input, $K$-output matchgates form a group under matrix multiplication. One can check that their construction carries over to show that the $2^{K}\times 2^{K}$ standard signatures of such matchgates likewise form a group, but unfortunately this is not enough to prove Theorem~\ref{thm:main_step}, as $G'$ alone is merely a pseudo-signature and may not be realizable as the standard signature of a $K$-input, $K$-output matchgate. That said, Theorem~\ref{thm:main_step} can still be proved with minor modifications to Li and Xia's approach.

We begin with a toy example motivating the notation in the previous paragraphs. Suppose that for each $i\in[K]$, there exists an edge of weight 1 such that the $i$-th external node in block $t$ and external node $q_i$ are both incident only to this edge. Note that in this case, $G'$ is a symmetric permutation matrix and thus equal to its own inverse. We can easily construct a recognizer $R$ out of $G$ for which $\underline{G}(t)\underline{R}(t) = I_{2^K}$ as follows. Remove all non-external nodes of $G$, as well as all edges incident to nodes outside of block $t$ and nodes $q_1,...,q_K$. For external node $i$ outside of block $t$ such that $i\not\in\{q_1,...,q_K\}$, if $\zeta_i = 0$, attach a distinct edge of weight 1 to node $i$ and designate the other endpoint of the edge as the $i$th external node of $R$; if $\zeta_i = 1$, attach a distinct path graph of length 2 consisting of two edges of weight 1, and denote the other endpoint of the path as the $i$th external node of $R$. By construction, in the $2^{K(n-1)}\times 2^{K}$ matrix $\underline{R}(t)$, the submatrix indexed by rows $q_1,...,q_K$ is equal to $G'$, and all other entries are zero. Because $G' = (G')^{-1}$, $\underline{G}(t)\underline{R}(t)= I_{2^K}$ as desired.

See Figure~\ref{fig:toy} for an example of this construction.

\begin{figure}[h]
\centering
	\includegraphics[height=0.6\textheight]{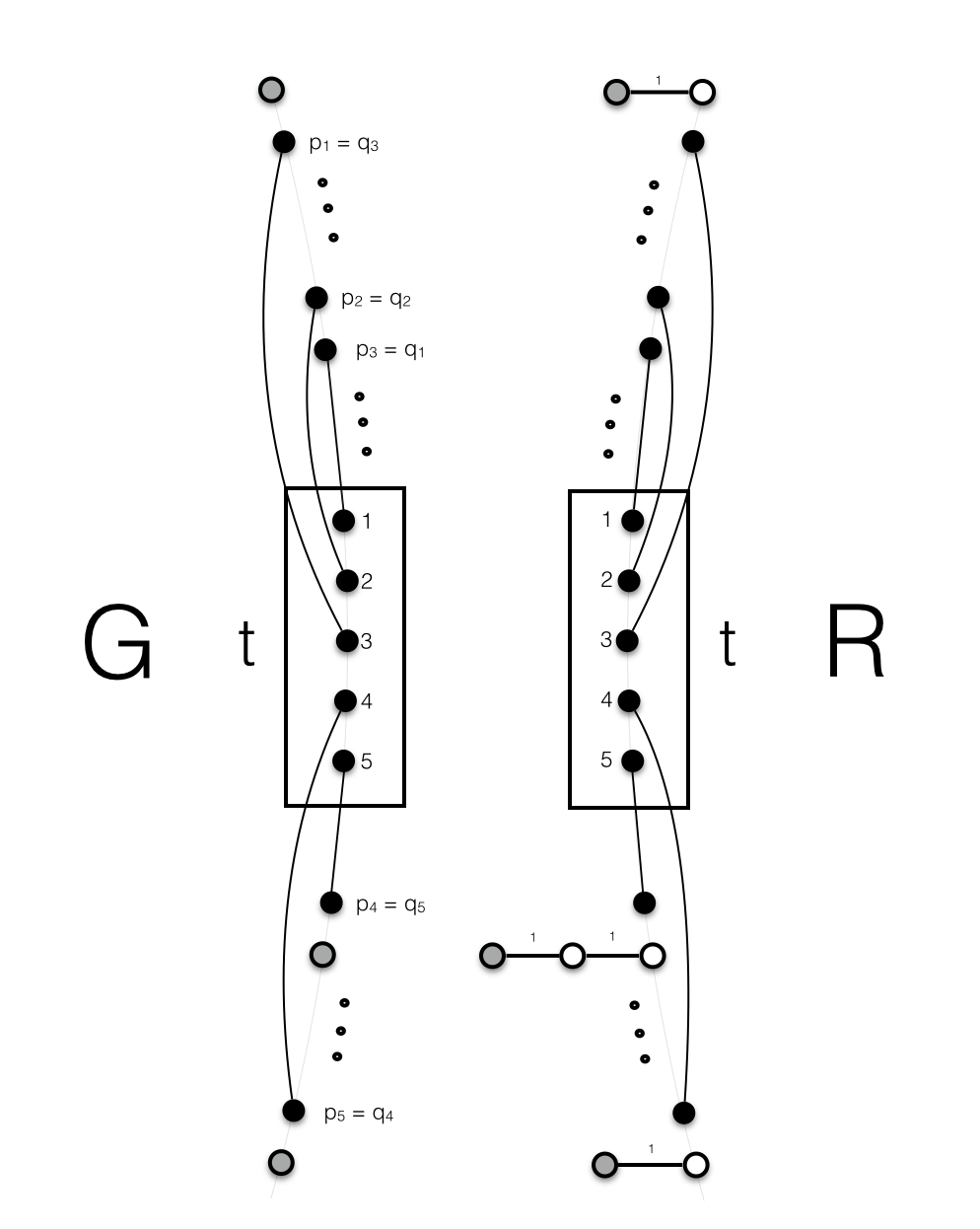}
	\caption{Toy example of $G$ reduced at $i = 1,...,K$. Here, $\ell = 5$, $K = 5$, $m = 3$. Black nodes in $G$ denote nodes $1,...,K$ in block $t$ and nodes $q_1,...,q_K$. External nodes of $G$ and $R$ are shown in black/gray.}
	\label{fig:toy}
\end{figure}

\begin{defn}
 	For $i\in[K]$, if $i\le m$ (resp. $i>m$), $G$ is \emph{reduced at $i$} if there exists an edge of weight 1 in $G$ connecting the $i$-th external node in block $t$ and external node $q_i$ such that these two nodes are both incident only to that edge.
 \end{defn}

To prove Theorem~\ref{thm:main_step}, it is enough to reduce to the special case of the toy example above where $\Gamma$ is realizable by a generator $G$ reduced at all $i\in[K]$. The rest of this section will be dedicated to proving the following:

\begin{lem}
	If $\Gamma$ is the standard signature of a generator of arity $Kn$ reduced at $1,...,i$, there exist nonsingular $K$-input, $K$-output transducers $L_1,...,L_r$ and $K(n-1)$-input, $K(n-1)$-output transducers $R_1,...,R_s$ such that $\underline{L_r}\cdots \underline{L_1}\cdot\Gamma\cdot \underline{R_1}\cdots \underline{R_s}$ is the standard signature of a generator reduced at $1,...,i+1$.\label{lem:transform}
\end{lem}

We first give a sufficient characterization of standard signatures of matchgates reduced at $1,...,i$ in terms of the entries of their standard signatures.

The following terminology is borrowed from \cite{groupprop}.

\begin{defn}
	Let $M$ be any $2^r\times 2^c$ matrix whose rows and columns are indexed by $\sigma\in\{0,1\}^r$ and $\tau\in\{0,1\}^{c}$ respectively. $\Gamma^{\sigma}_{\tau}$ is an \emph{edge entry of $M$} iff $r+c-2\le \wt(\sigma) + \wt(\tau) < r+c$.
\end{defn}

\begin{lem}
	$\Gamma$ is the standard signature of a generator $G$ that is reduced at $i$ if $\Gamma$ satisfies the following:

	\begin{enumerate}
	 	\item $(G')^{1^{K}}_{1^{K}} = (G')^{1^{K}\oplus e_{i}}_{1^{K}\oplus e_{q_i}} = 1$.

	 	\item $(G')^{\sigma}_{\tau} = 0$ for all other edge entries of $G'$ for which $\sigma\in\{1^K,1^K\oplus e_i\}$ or $\tau\in\{1^K,1^K\oplus e_{q_i}\}$.
	 \end{enumerate}
	 \label{lem:red_entries}
\end{lem}

To show this, it suffices to prove the following useful consequence of the matchgate identities, first observed in \cite{caichoudharylu} and translated below to our setting of standard signatures in matrix form.

\begin{lem}[Theorem 4.2, \cite{caichoudharylu}]
	If $(G')^{1^K}_{1^K}\neq 0$, the entries of $G'$ are uniquely determined by $(G')^{1^K}_{1^K}$ and the edge entries of $G'$.
	\label{lem:sig_dep}
\end{lem}

\begin{proof}
	Assume that these entries uniquely determine all entries $\Gamma^{\sigma}_{\tau}$ for which $\wt(\sigma)+\wt(\tau)\ge m$ for some $m\le n - 2$. We proceed by downward induction on $m$ (by the parity condition, if $m$ is even, the case of $m+1$ follows immediately from that of $m$). Then for $\sigma,\tau\in\{0,1\}^K$ such that $\wt(\sigma)+\wt(\tau) = m - 2$, apply \eqref{eq:mgis} from Theorem~\ref{thm:mgis} to $\sigma := \sigma$, $\zeta := \tau$, $\tau := 1^K$, and $\eta := 1^K$. One can check that the resulting identity consists of $(G')^{\sigma}_{\tau}\cdot(G')^{1^K}_{1^K}$ and terms which have already been uniquely determined by the inductive hypothesis, so because $(G')^{1^K}_{1^K}\neq 0$, we conclude that $(G')^{\sigma}_{\tau}$ is also uniquely determined.
\end{proof}

\begin{obs}
	If $\Gamma$ is the standard signature of a generator reduced at $i$, $(G')^{\sigma\oplus e_i}_{\tau\oplus e_{q_i}} = (G')^{\sigma}_{\tau}$, and if $\sigma_i\neq\tau_{q_i}$, $(G')^{\sigma}_{\tau} = 0$.\label{obs:block}
\end{obs}

\begin{proof}
	By hypothesis, external node $i$ of block $t$ and external node $q_i$ are only connected to each other. If $\sigma_i\neq\tau_{q_i}$, $(G')^{\sigma}_{\tau}$ counts the number of perfect matchings of $\Gamma$ where, among other conditions, exactly one of these two nodes is removed, and no such matching exists. On the other hand, if $\sigma_i = \tau_{q_i}$, then $(G')^{\sigma}_{\tau}$ and $(G')^{\sigma\oplus e_i}_{\tau\oplus e_{q_i}}$ count the number of perfect matchings in which, among other conditions, both of these two nodes are removed, or neither is. The number of perfect matchings in either scenario is the same.
\end{proof}

Let $G'_i$ be the $2^{K-i}\times 2^{K-i}$ submatrix of $G'$ consisting of entries $(G')^{\sigma}_{\tau}$ for which $\sigma_j = 0$ for $j = 1,...,i$ and $\tau_j = 0$ for $j = q_1,...,q_i$. If the row and column indices of $G'$ are of the form $i_1\cdots i_K$ and $i_{p_1}\cdots i_{p_K}$ respectively, the row and column indices of $G'_i$ are of the form $i_{i+1}\cdots i_K$ and $i_{p_{i+1}}\cdots i_{p_K}$ respectively. When referring to the row (resp. column) of $G'$ containing a row $i_{i+1}\cdots i_K$ (resp. column $i_{p_{i+1}}\cdots i_{p_K}$) of $G'_i$, we use the notation $0^i\circ i_{i+1}\cdots i_K$ (resp. $0^i\circ i_{p_{i+1}}\cdots i_{p_K}$) to denote its index in $G'$. For example, column $0^i\circ 1^{K-i}$ of $G'$ is the column of $G'$ indexed by $\sigma\in\{0,1\}^K$ for which $\sigma_{q_1} = \cdots = \sigma_{q_i} = 0$ and $\sigma_{q_{i+1}} = \cdots = \sigma_{q_K} = 1$, and this contains column $1^{K-i}$ of $G'_i$.

\begin{cor}
	If $\Gamma$ is the standard signature of a generator that is reduced at $1,...,i$, then $\Gamma$ is the standard signature of a generator reduced at $1,...,i+1$ if $\Gamma$ satisfies the following:

	\begin{enumerate}
	 	\item $(G'_{i})^{1^{K-i}}_{1^{K-i}} = (G'_{i})^{1^{K-i}\oplus e_{i+1}}_{1^{K-i}\oplus e_{q_{i+1}}} = 1$
	 	\item $(G'_i)^{\sigma}_{\tau} = 0$ for all other edge entries of $G'_i$ such that $\sigma\in\{1^{K-i},1^{K-i}\oplus e_{i+1}\}$ or $\tau\in\{1^{K-i},1^{K-i}\oplus e_{q_{i+1}}\}$.
	 \end{enumerate} 
\end{cor}

\begin{proof}
	Apply Observation~\ref{obs:block} to each of $1,...,i$, and invoke Lemma~\ref{lem:red_entries}.
\end{proof}

\begin{proof}[Proof of Lemma~\ref{lem:transform}]
	We execute the transformation $\Gamma =\Gamma^{(0)} \Rightarrow\Gamma^{(1)}\Rightarrow\Gamma^{(2)}\Rightarrow\Gamma^{(3)}\Rightarrow\Gamma^{(4)}$ outlined below.

	\begin{enumerate}
		\item ($\Gamma^{(0)}\Rightarrow \Gamma^{(1)}$): Turn the entry indexed by $(1^{K-i},1^{K-i})$ in $G'_i$ to 1.
		\item ($\Gamma^{(1)}\Rightarrow \Gamma^{(2)}$): Turn edge entries of $G'_i$ in row or column $1^{K-i}$ to 0.
		\item ($\Gamma^{(2)}\Rightarrow \Gamma^{(3)}$): Turn entry $(1^{K-i}\oplus e_{i+1},1^{K-i}\oplus e_{q_{i+1}})$ in $G'_i$ to 1.
		\item ($\Gamma^{(3)}\Rightarrow \Gamma^{(4)}$): Turn all other edge entries in $G'_i$ in row $1^{K-i}\oplus e_{i+1}$ or column $1^{K-i}\oplus e_{q_{i+1}}$ to zero. 
	\end{enumerate}

	We need not care what these transformations do to entries outside of $G'$, but we must ensure they preserve the fact that $\Gamma$ is the standard signature of a generator reduced at $1,...,i$. To do this, for each matrix $M$ by which we left- or right-multiply $\Gamma$, if $\sigma$ does not index a row (resp. column) of $G'_i$, the only nonzero entry of $M$ in row (resp. column) $\sigma$ will be 1 in column (resp. row) $\sigma$.

	In each step $j$, we will for convenience refer to $\Gamma^{(j-1)}$ as $\Gamma$.

	\begin{figure}[h]
		\centering
		\subfloat[][$C_j$]{\includegraphics[height=2.5in]{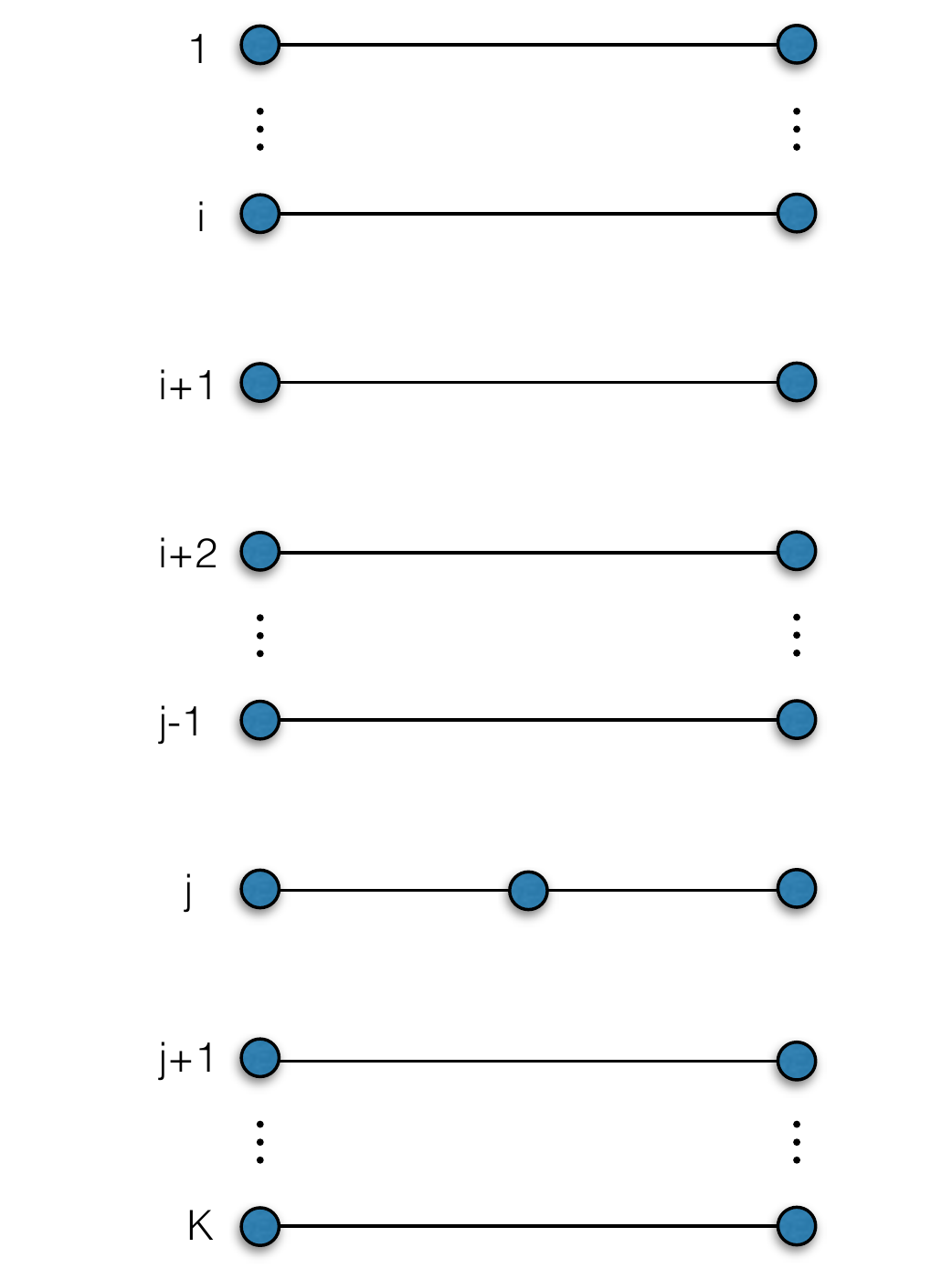}\label{fig:1}}\qquad
		\subfloat[][$L^{j,k}_3$]{\includegraphics[height=2.5in]{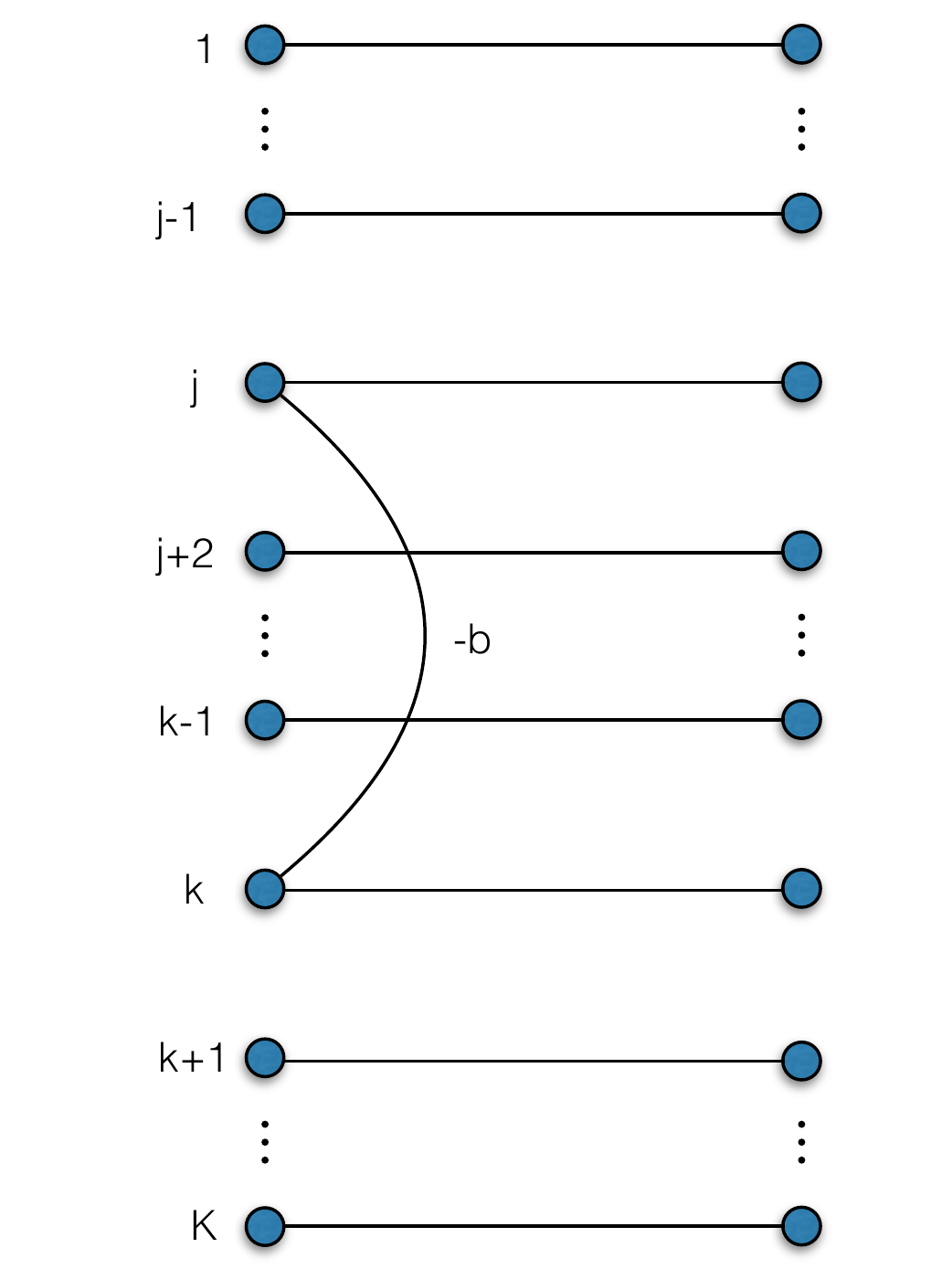}\label{fig:2}}\qquad
		\\
		\vspace{1.0cm}
		\subfloat[][$L_4$]{\includegraphics[height=2.5in]{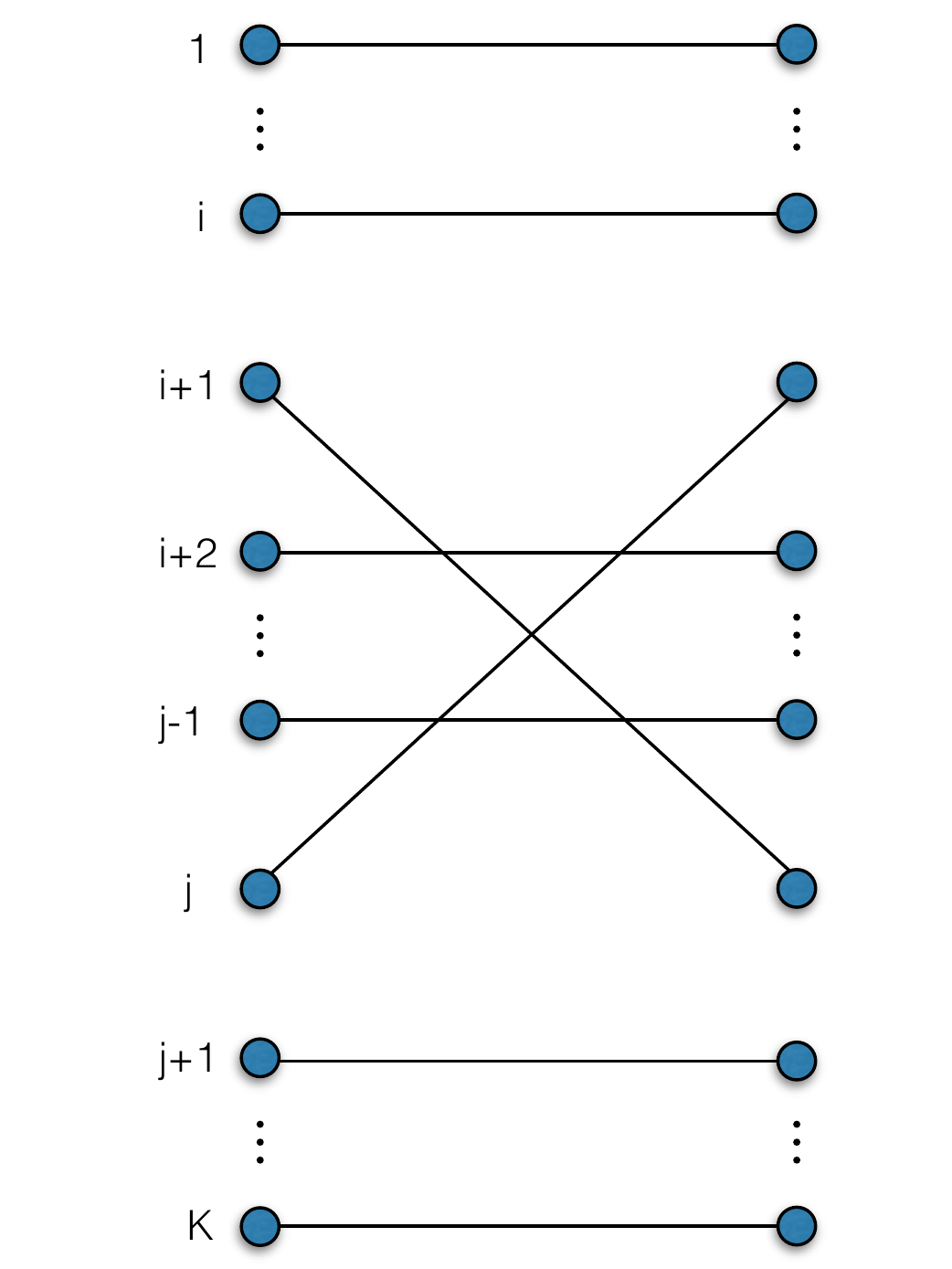}\label{fig:3}}\qquad
		\subfloat[][$L^j_5$]{\includegraphics[height=2.5in]{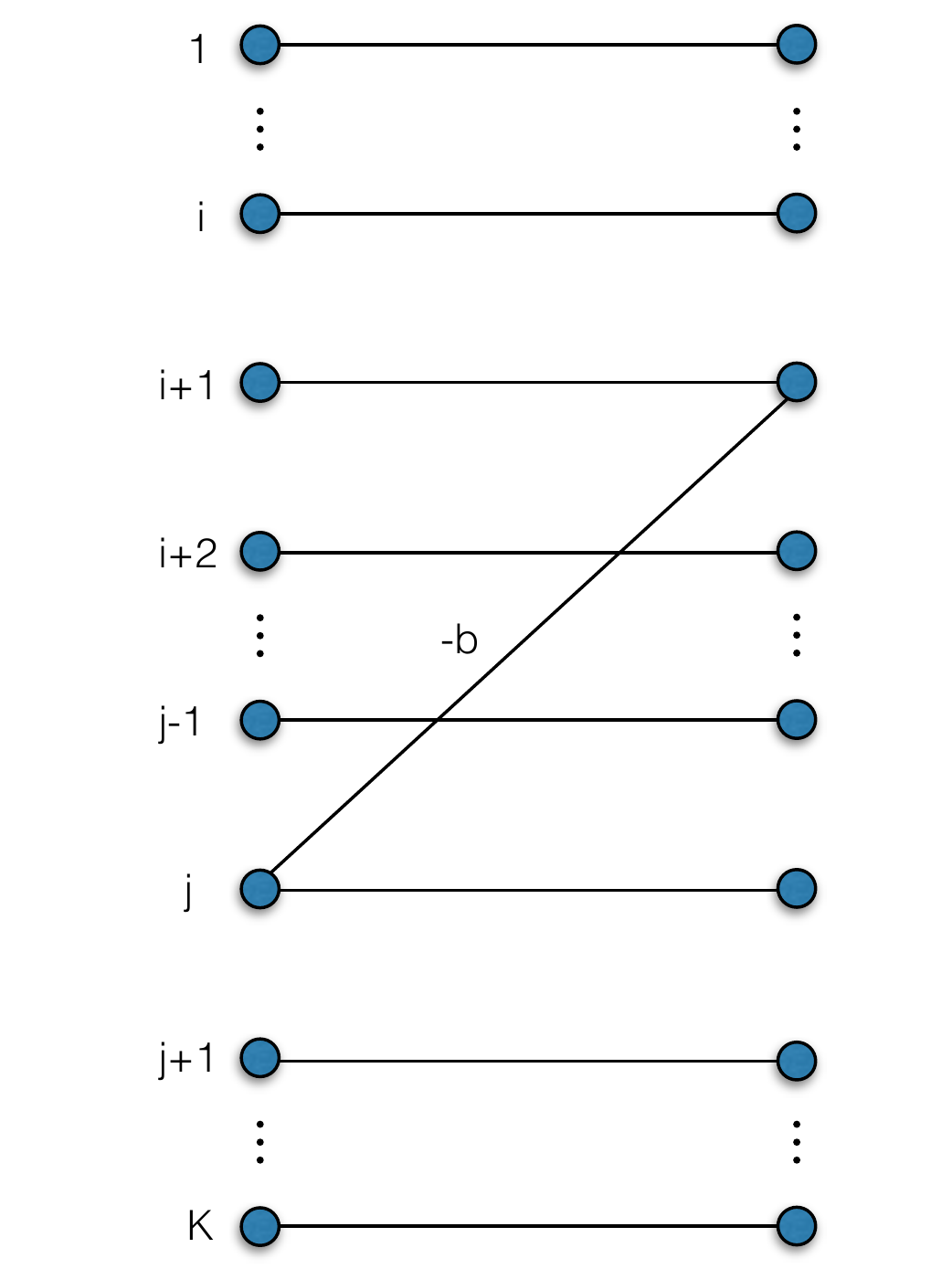}\label{fig:4}}
		\caption{Transducers realizing row/column operations in steps 1-4}
	\end{figure}

	\begin{step}
		($\Gamma^{(0)}\Rightarrow \Gamma^{(1)}$): Turn the entry indexed by $(1^{K-i},1^{K-i})$ in $G'_i$ to 1.
	\end{step}

	We first show how to move a nonzero entry $c:=(G'_i)^{\sigma^*}_{\tau^*}$ of $G'_i$ into entry $(1^{K-i},1^{K-i})$ of $G'_i$.

	For each $j$ for which $i<j\le K$, we would like a $2^{K}\times 2^K$ standard signature $C_j$ such that left-multiplication of $\Gamma$ by $C_j$ interchanges row $\sigma$ in $G'$ with $\sigma\oplus e_j$ for all $\sigma\in\{0,1\}^K$, and a $2^{K(n-1)}\times 2^{K(n-1)}$ standard signature $D_j$ such that right-multiplication of $\Gamma$ by $D_j$ interchanges column $\tau$ in $G'$ with $\tau\oplus e_{q_{j}}$. We could then define $L_1 = \prod_{j:\sigma^*_j = 0}C_j$ and $R_1 = \prod_{j:\tau^*_j = 0}D_j$, and $L_1\cdot\Gamma R_1$ would have nonzero entry $c$ at index $(1^{K-i},1^{K-i})$ of $G'_i$ and still be the standard signature of a matchgate reduced at $1,...,i$.

	$C_j$ (resp. $D_j)$ is the permutation matrix whose only nonzero entry in row $\sigma\in\{0,1\}^K$ (resp. $\sigma\in\{0,1\}^{K(n-1)}$) is 1 in column $\sigma\oplus e_j$ (resp. column $\sigma\oplus e_{q_j}$) if $G'_i$ contains entries from $\Gamma^{\sigma}$ (resp. $\Gamma_{\sigma}$), and 1 in column $\sigma$ otherwise. $C_j$ and $D_j$ are certainly nonsingular.

	To construct the $K$-input, $K$-output transducer realizing $C_j$ as a standard signature, begin with a $(K,K)$-bipartite graph where for every $\nu\neq j$, left node $\nu$ and right node $\nu$ are connected by an edge of weight 1. Add an extra vertex between left node $j$ and right node $j$, and draw a path of length two connecting these three vertices, where both edges of the path have weight 1. This construction is shown in Figure~\ref{fig:1}.

	The $K(n-1)$-input, $K(n-1)$-output transducer realizing $D_j$ as a standard signature is similarly constructed, the only difference being that the bipartite graph has left and right vertex sets of size $K(n-1)$, and the path of length two is drawn between left node $q_j$ and right node $q_j$.

	Next, we want to scale all of the entries of $L_1\Gamma R_1$ by a factor of $1/c$, so define $L_2(c)$ to be the diagonal matrix whose entry at index $1^{K-i}$ is $1/c$ and whose entries at all other indices are 1. Obviously $L_2(c)$ is nonsingular and satisfies the matchgate identities \eqref{eq:mgis}. We take $\Gamma^{(1)} = L_2(c)L_1\Gamma^{(0)} R_1$.


	\begin{step}
		($\Gamma^{(1)}\Rightarrow \Gamma^{(2)}$): Turn edge entries of $G'_i$ in row or column $1^{K-i}$ to 0.
	\end{step}

	We demonstrate how to do this for edge entries in column $1^{K-i}$. Firstly, edge entries $(1^{K-i}\oplus e_j,1^{K-i})$ in $G'_i$ are already zero by the parity condition.\footnote{As characters of general matchgates with omittable nodes do not satisfy the parity condition necessarily, the proof of the group property in \cite{groupprop} requires an extra construction to turn edge entries $(1^{K-i}\oplus e_j,1^{K-i})$ to zero. This is an instance where our proof of the group property for signatures is easier than that for characters.} To set each of the remaining edge entries in this column to zero, we will proceed in reverse lexicographic order over the rows $1^{K-i}\oplus e_j\oplus e_k$ of $G'_i$ and at each step left-multiply $\Gamma$ by a matrix $L^{j,k}_3$ which corresponds in $\Gamma$ to subtracting $b:= (G'_i)^{1^{K-i}\oplus e_j\oplus e_k}_{1^{K-i}}$ times row $0^i\circ 1^{K-i}$ of $G'$ from row $0^i\circ(1^{K-i}\oplus e_j\oplus e_k)$ of $G'$.

	$L^{j,k}_3$ must be a matrix whose nonzero entries include diagonal entries equal to 1 and entry $(0^i\circ(1^{K-i}\oplus e_j\oplus e_k), 0^i\circ 1^{K-i})$ equal to $-b$. A standard signature satisfying these conditions can be realized by the following matchgate: construct a $(K,K)$-bipartite graph in which left node $\nu$ and right node $\nu$ are connected by an edge of weight 1 for all $\nu$, and left nodes $j$ and $k$ are also connected by an edge of weight $-b$. The standard signature of this matchgate is nonsingular. The construction is shown in Figure~\ref{fig:2}.

	$L^{j,k}_3$ additionally has nonzero entries $(\sigma\oplus e_j\oplus e_k, \sigma)$ equal to $-b$ for all $\sigma\in\{0,1\}^{K}$, i.e. left-multiplication by $L^{j,k}_3$ corresponds to subtracting $b$ times row $\sigma$ from row $\sigma\oplus e_j\oplus e_k$. These extraneous side effects do not however affect any of the progress we've made as the rows $1^{K-i}\oplus e_j\oplus e_k$ of $G'_i$ are taken in reverse lexicographic order.

	The only issue is that the matchgate we have constructed is not necessarily planar. But by Lemma~\ref{lem:crossover1} in Appendix~\ref{app:crossover}, there exists a planar matchgate with standard signature equal to $L^{j,k}_3$, except at nonzero off-diagonal entries other than $(0^i\circ (1^{K-i}\oplus e_j\oplus e_k), 0^i\circ 1^{K-i})$, where it may differ from $L^{j,k}_3$ by a factor of $-1$, but by the reasoning in the previous paragraph, this does not matter. Denote the standard signature of this planar matchgate by $L'^{j,k}_3$.

	To achieve step 2 for edge entries in \emph{rows} $1^{K-i}$ as well, we can define matrices $R'^{j,k}_3$ analogously. We can thus take $\Gamma^{(2)} = \left(\prod_{i+2\le j,k\le K} L'^{j,k}_3\right)\cdot\Gamma^{(1)}\cdot\left(\prod_{i+2\le j,k\le K}R'^{j,k}_3\right)$, where the indexing in the products respects the abovementioned reverse lexicographic order.

	\begin{step}
		($\Gamma^{(2)}\Rightarrow \Gamma^{(3)}$): Turn entry $(1^{K-i}\oplus e_{i+1},1^{K-i}\oplus e_{q_{i+1}})$ in $G'_i$ to 1.
	\end{step}

	Note that $c' := (G'_i)^{1^{K-i}\oplus e_{j}}_{1^{K-i}\oplus e_{q_k}}$ must be nonzero for some $j,k\in[K]$ or else $(G'_i)$ is singular. As in Step 1, we will first left-multiply $\Gamma$ by some $L_4$ to move this nonzero entry to row $1^{K-i}\oplus e_{{i+1}}$ and then right-multiply by some $R_4$ to move it to column $1^{K-i}\oplus e_{q_{i+1}}$. Unfortunately, multiplying by $C_j$ or $D_j$ defined in Step 1 would interfere with the progress made so far in Steps 1 and 2.

	Instead, $L_4$ must be a matrix whose only nonzero entry in row (resp. column) $0^i\circ(1^{K-i}\oplus e_j)$ is 1 in column (resp. row) $0^i\circ(1^{K-i}\oplus e_{i+1})$, and whose only nonzero entry in row (resp. column) $0^i\circ(1^{K-i}\oplus e_{i+1})$ is 1 in column (resp. row) $0^i\circ(1^{K-i}\oplus e_j)$. A signature satisfying these conditions can be realized by the following matchgate: construct a $(K,K)$-bipartite graph in which left node $\nu$ and right node $\nu$ are connected by an edge of weight 1 for all $\nu\neq j,i+1$. Connect left node $j$ to right node $i+1$ and left node $i+1$ to right node $j$ with edges of weight 1. The standard signature $L_4$ of this is nonsingular. The construction is shown in Figure~\ref{fig:3}.

	$L_4$ also has nonzero entries $(i_1\cdots i_j\cdots i_{i+1}\cdots i_K,i_1\cdots i_{i+1}\cdots i_j\cdots i_K)$, so left-multiplication by $L_4$ corresponds to switching row $0^i\circ(i_1\cdots i_{i+1}i_{i+2}\cdots i_j\cdots i_{K})$ with row $0^i\circ(i_1\cdots i_ji_{i+2}\cdots i_{i+1}\cdots i_{K})$ for all $i_{i+1},...,i_{K}\in\{0,1\}$. Multiplication by $L_4$ affects neither the progress we've made on entry $(1^{K-i},1^{K-i})$ of $G'_i$ because all bits in the row and column indices are equal, nor the progress on the edge entries in row $1^{K-i}$ and column $1^{K-i}$ of $G'_i$ because these get swapped with each other and were already all zero, keeping them equal to zero.

	As before, the only issue is that the matchgate constructed is not planar. But by Lemma~\ref{lem:crossover2} in Appendix~\ref{app:crossover}, there exists a planar matchgate whose standard signature agrees with $L_4$ at every entry up to sign. By the above, that the nonzero entries other than $(0^i\circ 1^{K-i}\oplus e_j,0^i\circ 1^{K-i}\oplus e_{i+1})$ and $(0^i\circ 1^{K-i}\oplus e_{i+1},0^i\circ 1^{K-i}\oplus e_j)$ may be $-1$ does not matter. Furthermore, if either entry $(0^i\circ 1^{K-i}\oplus e_j,0^i\circ 1^{K-i}\oplus e_{i+1})$ or $(0^i\circ 1^{K-i}\oplus e_{i+1},0^i\circ 1^{K-i}\oplus e_j)$ were $-1$, at worst we may eventually need to replace $c'$ with $-c'$, but the argument still holds. Denote the standard signature of this planar matchgate by $L'_4$. We can analogously define $R'_4$.

	Next, we scale entry $(1^{K-i}\oplus e_{i+1},1^{K-i}\oplus e_{q_{i+1}})$ of $G'_i$ by a factor of $1/c'$, so we take $\Gamma^{(3)} = L_2(c')L'_4\Gamma^{(2)} R'_4$.

	\begin{step}
		($\Gamma^{(3)}\Rightarrow \Gamma^{(4)}$): Turn all other edge entries in $G'_i$ in row $1^{K-i}\oplus e_{i+1}$ or column $1^{K-i}\oplus e_{q_{i+1}}$ to zero.
	\end{step}

	We demonstrate how to do this for edge entries in column $1^{K-i}\oplus e_{q_{i+1}}$ of $G'_i$. To set each of the edge entries in this column to zero, we will proceed in reverse lexicographic order over the row indices $1^{K-i}\oplus e_j$ of $G'_i$ and at each step left-multiply $\Gamma$ by a matrix $L^j_5$ which subtracts $b:= (G'_i)^{1^{K-i}\oplus e_j}_{1^{K-i}\oplus e_{q_{i+1}}}$ times row $0^i\circ 1^{K-i}$ of $G'$ from row $0^i\circ (1^{K-i}\oplus e_j)$ of $G'$.

	$L^{j}_5$ is a matrix whose nonzero entries include diagonal entries equal to 1 and entry $(0^i\circ (1^{K-i}\oplus e_j), 0^i\circ (1^{K-i}\oplus e_{i+1}))$ equal to $-b$. A signature satisfying these conditions can be realized by the following matchgate: construct a $(K,K)$-bipartite graph in which left node $\nu$ and right node $\nu$ are connected by an edge of weight 1 for all $\nu$, and connect left node $j$ to right node $i+1$ by an edge of weight $-b$. The standard signature of this matchgate is nonsingular. The construction is shown in Figure~\ref{fig:4}.

	$L^{j}_5$ additionally has nonzero entries $(0^i\circ(\sigma\oplus e_j\oplus e_{i+1}), 0^i\circ \sigma)$ equal to $-b$ for all $\sigma\in\{0,1\}^{K-i}$ such that $\sigma_{j} = 1$ and $\sigma_{i+1} = 0$, i.e. for all such $\sigma$, left-multiplication by $L^{j}_5$ corresponds to subtracting $b$ times row $0^i\circ \sigma$ of $G'$ from row $0^i\circ \sigma\oplus e_j\oplus e_{i+1}$ of $G'$. As before, these extraneous side effects do not affect the progress we've made in this step because the edge entries in column $1^{K-i}\oplus e_{q_{i+1}}$ of $G'_i$ are being taken in reverse lexicographic order. They also certainly do not affect row $1^{K-i}$, nor do they affect column $1^{K-i}$ as $1^{K-i}$ does not satisfy the above criteria for $\sigma$.

	Again, the issue is that the matchgate constructed is not planar, but by Lemma~\ref{lem:crossover3} in Appendix~\ref{app:crossover}, there exists a planar matchgate with standard signature equal to $L^j_5$ except possibly at the nonzero off-diagonal entries other than $(0^i\circ (1^{K-i}\oplus e_j), 0^i\circ (1^{K-i}\oplus e_{i+1}))$, where it may differ by a factor of $-1$. Denote the standard signature of this planar matchgate by $L'^j_5$.

	To achieve step 4 for edge entries in rows $1^{K-i}\oplus e_{i+1}$ as well, we can define matrices $R'^j_5$ analogously. We can thus take $\Gamma^{(4)} = \left(\prod_{i+2\le j\le K} L'^{j}_5\right)\cdot\Gamma^{(3)}\cdot\left(\prod_{i+2\le j\le K}R'^{j}_5\right)$, where the products respect the abovementioned reverse lexicographic order.
\end{proof}

\section{Reducing to Domain Size $2^K$}
\label{sec:fuyang}

In this section we use Theorem \ref{thm:rigid} to reduce proving a basis collapse theorem over all domain sizes to proving one over domain sizes $2^K$. The result we will prove is the following generalization of the main result in \cite{fuyang} whose strategy we follow.

\begin{thm}
	Suppose Theorem~\ref{thm:main} has been proven for domain size $r$. If recognizer signatures $R_1,...,R_a$ and generator signatures $G_1,...,G_b$ on domain size $k > r$ belonging to matchgrid $\Omega$ are simultaneously realizable on a $2^{\ell}\times k$ basis $M$ of rank $r$ and $R_1$ is of full rank, then there exists a basis $M'$ of size at most $\lfloor\log_2 r\rfloor$ on which they are simultaneously realizable.\label{thm:fuyanggen}
\end{thm}

We'll need some preliminaries before we can prove this. Express $M$ as $\begin{pmatrix} \alpha_1 & \alpha_2 & \cdots & \alpha_k\end{pmatrix}$ where each $\alpha_i$ is a $2^{\ell}$-dimensional column vector. Let $i_1,...,i_r\in[k]$ be column indices of $M$ for which $M^{i_1\cdots i_r} := \begin{pmatrix} \alpha_{i_1} & \alpha_{i_2} & \cdots & \alpha_{i_r}\end{pmatrix}$ is of full rank. Define \emph{sub-signature} $R^{i_1\cdots i_r}$ to consist of entries $(R_{j_1\cdots j_n})$ of $R$ ranging over all $j_1,...,j_n\in\{i_1,...,i_r\}\subset[k]$. We can define the sub-signature $G^{i_1\cdots i_r}$ of a generator analogously. Equivalently, \begin{equation}R^{i_1\cdots i_r} = \underline{R}(M^{i_1\cdots i_r})^{\tensor n}\label{eq:smaller}\end{equation}

\begin{lem}
	For a recognizer $R$ realizable on basis $M$, if there exists $t$ for which $\rank(R(t))\ge r$, then $\rank(R^{i_1\cdots i_r}(t)) = r$.\label{lem:subsig}
\end{lem}

\begin{proof}
	By Lemma~\ref{lem:relunderline2}, $R(t) = (M^T)^{\tensor(n-1)}\underline{R}(t)M$, so $\rank(\underline{R}(t))\ge r$. By \eqref{eq:smaller} and Lemma~\ref{lem:relunderline}, $R^{i_1\cdots i_r}(t) = ((M^{i_1\cdots i_r})^T)^{\tensor(n-1)}\underline{R}(t)M^{i_1\cdots i_r}$, so $\rank(R^{i_1\cdots i_r}(t)) = r$.
\end{proof}

For such a recognizer $R$, define for each $w\in[k]$ a $nk^{n-1}$-dimensional column vector $b_w$ by \begin{equation}
	b_w = \begin{pmatrix}
		R_{w 1\cdots 11} & \cdots & R_{wk\cdots kk} & R_{1 w\cdots 11} & \cdots & R_{k w\cdots kk} & \cdots & R_{11\cdots 1w} & \cdots & R_{kk\cdots kw}
	\end{pmatrix}^T\label{eq:bw}
\end{equation}
and define $A_{i_1\cdots i_r}$ to be the $nk^{(n-1)}\times r$ matrix whose $j$th column is $b_{i_j}$.

\begin{obs}
	$\rank(A_{i_1\cdots i_r}) = r$.\label{obs:arank}
\end{obs}

\begin{proof}
	$R^{i_1\cdots i_r}(t)$ is a submatrix of $A_{i_1\cdots i_r}$ and already has rank $r$ by Lemma~\ref{lem:subsig}.
\end{proof}

\begin{obs}
	As $M$ has rank $r$, every column $\alpha_w$ can be expressed as a linear combination $\sum^r_{j=1}x^{i_j}_w\alpha_{i_j}$.\label{obs:linearcombo}
\end{obs}

Denote the $r\times k$ matrix $\left(x_w^{i_j}\right)$ of these coefficients by $X_{i_1\cdots i_r}$.

\begin{lem}
	For each $w\in[k]$, $A_{i_1\cdots i_r}X = b_w$ has the unique solution $X = \begin{pmatrix} x^{i_1}_w & \cdots & x^{i_r}_w \end{pmatrix}^T$.\label{lem:uniquesol}
\end{lem}

\begin{proof}
	A solution for $X$ exists and is unique because $\rank(A_{i_1\cdots i_r})= r$ by Observation~\ref{obs:arank}. To check that the purported solution for $X$ is correct, pick any entry $R_{j_1\cdots j_{t-1}wj_t\cdots j_n}$ of $b_w$. By definition of recognizer signatures, \begin{align*}
		R_{j_1\cdots j_{t-1}wj_t\cdots j_n} &= \langle\underline{R},\alpha_{j_1}\tensor\cdots\tensor\alpha_{j_{t-1}}\tensor\alpha_w\tensor\alpha_{j_{t+1}}\tensor\cdots\tensor\alpha_{j_n}\rangle \\
		&= \langle\underline{R},\alpha_{j_1}\tensor\cdots\tensor\alpha_{j_{t-1}}\tensor\left(\sum^r_{j=1}x^{i_j}_w\alpha_{i_j}\right)\tensor\alpha_{j_{t+1}}\tensor\cdots\tensor\alpha_{j_n}\rangle \\
		&= \sum^{r}_{j=1}x^{i_j}_w\cdot\left\langle\underline{R},\alpha_{j_1}\tensor\cdots\tensor\alpha_{j_{t-1}}\tensor\alpha_{i_j}\tensor\alpha_{j_{t+1}}\tensor\cdots\tensor\alpha_{j_n}\right\rangle \\
		&= \sum^{r}_{j=1}x^{i_j}_wR_{j_1\cdots j_{t-1}i_jj_{t+1}\cdots j_n}.
	\end{align*} Here $\langle\cdot,\cdot\rangle$ denotes the inner product.
\end{proof}

The content of Lemma~\ref{lem:uniquesol} is that to any such $R$ we can get a matrix $X_{i_1\cdots i_r}$ without needing to know the actual basis $M$ realizing $R$.

\begin{lem}
	If $\rank(R_1(t))\ge r$ for some $t$, then recognizers $R_1,...,R_a$ are simultaneously realizable on some basis of rank $r$ iff the following conditions hold:

	\begin{enumerate}
		\item $\rank(R^{i_1\cdots i_r}_1(t)) = r$ for some $i_1,...,i_r\in[k]$.
		\item There exists a unique $r\times k$ matrix $X_{i_1\cdots i_r} = (x^{i_j}_w)$ such that $A_{i_1\cdots i_r}X = b_w$ has the solution $X = \begin{pmatrix} x^{i_1}_w & \cdots & x^{i_r}_w \end{pmatrix}^T$ for each $w\in[k]$.
		\item There exists a $2^{\ell}\times r$ basis $M_{(r)}$ such that the $R^{i_1\cdots i_r}_j$ are simultaneously realizable on $M_{(r)}$ for all $j\in[a]$.
		\item $R_j = R^{i_1\cdots i_r}_jX^{\tensor n}_{i_1\cdots i_r}$ for all $j\in[a]$.
	\end{enumerate}\label{lem:conds}
\end{lem}

\begin{proof}
	Suppose $R_1,...,R_a$ are simultaneously realizable are some basis $M$. Conditions 1 and 2 follow from Lemma~\ref{lem:subsig} and Lemma~\ref{lem:uniquesol} respectively. Take $M_{(r)}$ to be $M^{i_1\cdots i_r}$, and condition 3 follows from the definition of sub-signature. By Observation~\ref{obs:linearcombo}, $X_{i_1\cdots i_r}$ satisfies $M_{(r)}X_{i_1\cdots i_r} = M$, so $R_j = \underline{R}_jM^{\tensor n} = \underline{R}_jM_{(r)}^{\tensor n}X^{\tensor n}_{i_1\cdots i_r} = R^{i_1\cdots i_r}_jX^{\tensor n}_{i_1\cdots i_r}$, and condition 4 follows.

	Conversely, suppose conditions 1-4 hold. Condition 3 tells us that there is some $M_{(r)}$ for which $R^{i_1\cdots i_r}_j = \underline{R}_j\left(M_{(r)}\right)^{\tensor n}$ for all $j\in[a]$. Then $$R_j = R^{i_1\cdots i_r}_jX^{\tensor n}_{i_1\cdots i_r} = \underline{R}_j(M_{(r)})^{\tensor n}X^{\tensor n}_{i_1\cdots i_r} = \underline{R}_j(M_{(r)} X_{i_1\cdots i_r})^{\tensor n},$$ so $R_1,...,R_a$ are simultaneously realizable on $M := M_{(r)}X_{i_1,...,i_r}$.
\end{proof}

\begin{thm}
	If recognizer signatures $R_1,...,R_a$ and generator signatures $G_1,...,G_b$ in matchgrid $\Omega$ are simultaneously realizable on a basis of rank $r$ and there exists $t$ for which $\rank(R_1(t))\ge r$, then there exist recognizer signatures $\check{R}_1,...,\check{R}_a$ and generator signatures $\check{G}_1,...,\check{G}_b$ in matchgrid $\Omega'$ over domain size $r$ that are simultaneously realizable on a $2^{\ell}\times r$ basis $M_{(r)}$. Furthermore, \begin{equation}
		\Holant(\Omega) = \Holant(\Omega')\label{eq:holants}.
	\end{equation}\label{thm:getcheck}
\end{thm}

\begin{proof}
	We first construct $\check{R}_1,...,\check{R}_a,\check{G}_1,...,\check{G}_b$. $X_{i_1,...,i_r}$ obtained from $R_1$ via Lemma~\ref{lem:uniquesol} has rank $r$, so let $X'_{i_1\cdots i_r}$ be the $k\times k$ invertible matrix for which $X_{i_1\cdots i_r}X'_{i_1\cdots i_r} = \begin{pmatrix}
		I_r \ | \ \textbf{0}_{r\times (k-r)}
	\end{pmatrix}$, where $I_r$ is the $r\times r$ identity matrix and $\textbf{0}_{r\times (k-r)}$ denotes the $r\times (k-r)$ matrix consisting solely of zeroes. For each $j\in[a]$, let $R'_j = R_j(X'_{i_1,...,i_r})$, and let $\check{R}_j$ be the sub-signature $(R'_j)^{1\cdots r}$. Likewise, for each $j\in[b]$, let $G'_j$ be defined by $G_j = (X'_{i_1\cdots i_r})^{\tensor n}G'_j$, and let $\check{G}_j$ be the sub-signature $(G'_j)^{1\cdots r}$.

	\begin{claim}
		For all $j$, $\check{R}_j = R^{i_1\cdots i_r}_j$ and $\check{G}_j = G^{i_1\cdots i_r}_j$.\label{claim:check}
	\end{claim}

	\begin{proof}
		We need to check that \begin{equation}
			\check{R}_j = \underline{R}_j(M^{i_1\cdots i_r})^{\tensor n}\label{eq:claim1}
		\end{equation} \begin{equation}(M^{i_1\cdots i_r})^{\tensor n}\check{G}_j = \underline{G}_j.\label{eq:claim2}\end{equation} Indeed, \begin{align*}
			R'_j &= \underline{R}_j M^{\tensor n}(X'_{i_1\cdots i_r})^{\tensor n} \\
			&= \underline{R}_j(M^{i_1\cdots i_r})^{\tensor n}X^{\tensor n}_{i_1\cdots i_r}(X'_{i_1\cdots i_r})^{\tensor n} \\
			&= \underline{R}_j \begin{pmatrix}
				M^{i_1\cdots i_r} \ | \ \textbf{0}_{r\times(k-r)}
			\end{pmatrix}^{\tensor n},
		\end{align*} proving \eqref{eq:claim1}. Similarly, \begin{align*}
			\underline{G}_j &= M^{\tensor n}G_j \\
			&= M^{\tensor n}(X'_{i_1\cdots i_r})^{\tensor n}G'_j \\
			&= \begin{pmatrix}M^{i_1\cdots i_r} \ | \ \textbf{0}_{r\times(k-r)}\end{pmatrix}^{\tensor n}G'_j,
		\end{align*} proving \eqref{eq:claim2}.
	\end{proof}

	We conclude that $\check{R}_1,...,\check{R}_a,\check{G}_1,...,\check{G}_b$ are simultaneously realizable on the basis $M_{(r)} := M^{i_1\cdots i_r}$.

	To check that the Holants agree, first note that if $R'_1,...,R'_a,G'_1,...,G'_b$ lie in a corresponding matchgrid $\Omega''$, $\Holant(\Omega) = \Holant(\Omega'')$ because we're just applying a basis change from $M$ to $MX'_{i_1\cdots i_r}$. And $\Holant(\Omega') = \Holant(\Omega'')$ because the operation of taking sub-signatures does not lose any information in this case, i.e. $(R'_j)_{\sigma} = 0$ for all $\sigma\in[k]^n\backslash[r]^n$.
\end{proof}

For the next two results, suppose recognizer signatures $R_1,...,R_a$ and generator signatures $G_1,...,G_b$ in matchgrid $\Omega$ are simultaneously realizable on a basis of rank $r$ and there exists $t$ for which $\rank(R_1(t))\ge r$.

\begin{thm}
	 If the recognizer signatures $\check{R}_1,...,\check{R}_a$ and generator signatures $\check{G}_1,...,\check{G}_b$ constructed in Theorem~\ref{thm:getcheck} are also simultaneously realizable on a $2^{\ell'}\times r$ basis $M'_{(r)}$ of rank $r$, then recognizer signatures $R_1,...,R_a$ and generator signatures $G_1,...,G_r$ are simultaneously realizable on the $2^{\ell'}\times k$ basis $M'_{(r)}X_{i_1\cdots i_r}$, where $X_{i_1\cdots i_r}$ is obtained from $R_1$ by Lemma~\ref{lem:uniquesol}.\label{thm:reversecheck}
\end{thm}

\begin{proof}
	 $$R_j = R^{i_1\cdots i_r}_jX^{\tensor n}_{i_1\cdots i_r} = \check{R}_jX^{\tensor n}_{i_1\cdots i_r} = \underline{R}(M'_{(r)}X_{i_1\cdots i_r})^{\tensor n},$$ where the first equality holds by condition 4 of Lemma~\ref{lem:conds}, the second by Claim~\ref{claim:check}, the third by definition of $M'_{(r)}$. Likewise, because $$\check{G}_j = \begin{pmatrix}
	 	 	I_r \ | \ \textbf{0}_{r\times(k-r)}
	 	 \end{pmatrix}G'_j = X_{i_1\cdots i_r}X'_{i_1\cdots i_r}G'_j,$$
 	 we have that $$\underline{G}_j = M^{\tensor n}_{(r)}\check{G}_j = (M_{(r)}X_{i_1\cdots i_r})^{\tensor n}X'^{\tensor n}_{i_1\cdots i_r}G'_j = (M_{(r)}X_{i_1\cdots i_r})^{\tensor n}G_j,$$ so we conclude that $R_1,...,R_a,G_1,...,G_b$ are indeed simultaneously realizable on $M_{(r)}X_{i_1\cdots i_r}$.
\end{proof}

We are now ready to prove Theorem~\ref{thm:fuyanggen}.

\begin{proof}[Proof of Theorem~\ref{thm:fuyanggen}]
	By Theorem~\ref{thm:getcheck}, signatures $\check{R}_1,...,\check{R}_a$ and $\check{G}_1,...,\check{G}_b$ on domain size $r$ are simultaneously realizable on a $2^{\ell}\times r$ basis.

	By definition, $\check{R}_1 = R^{i_1\cdots i_r}_1$, and because $R_1$ was assumed to be full-rank, Lemma~\ref{lem:subsig} tells us that $\check{R}_1$ is full-rank. Then by the hypothesis that Theorem~\ref{thm:main} has already been proven for domain size $r$, there exists a $2^{\lfloor\log_2 r\rfloor}\times r$ basis $M'_{(r)}$ on which $\check{R}_1,...,\check{R}_a$ and $\check{G}_1,...,\check{G}_b$ are simultaneously realizable. By Theorem~\ref{thm:reversecheck}, $R_1,...,R_a$ and $G_1,...,G_b$ are simultaneously realizable on $2^{\lfloor\log_2 r\rfloor}\times k$ basis $M' := M'_{(r)}X_{i_1\cdots i_r}$.
\end{proof}

By Corollary~\ref{cor:red} and Theorem~\ref{thm:fuyanggen}, it remains to prove collapse theorems for holographic algorithms on domain sizes $k = 2^K$ and over bases of full rank, after which we get the following corollary.

\begin{cor}
	Any holographic algorithm on a basis of size $\ell$ and domain size $k$ not a power of 2 which uses at least one generator signature of full rank can be simulated on a basis of size at most $2^{\lfloor\log_2 k\rfloor}$.
\end{cor}

\section{Collapse Theorem For Domain Size $2^K$}
\label{sec:collapse-final}

The following is a direct generalization of the argument from Section 5.3 of \cite{caifu}, but we include it for completeness. We will take $G$ to be a generator signature of full rank on domain size $k=2^K$, basis $M$ to be a $2^{\ell}\times 2^K$ matrix of rank $2^K$, and $\underline{G} = M^{\tensor n}G$ to be the corresponding standard signature of arity $n\ell$. By Theorem~\ref{thm:corsub} applied to the transpose of $\underline{G}(t)$, there exists a cluster $Z = s+\{e_{p_1},...,e_{p_K}\}$ of rows of full rank in $\underline{G}(t)$. Denote by $M^Z$ the submatrix of $M$ consisting of rows with indices in $Z$.

\begin{lem}
	$M^Z$ is invertible.
\end{lem}

\begin{proof}
	The $(k,n\ell)$ cluster submatrix of $\underline{G}(t)$ of full rank whose existence is guaranteed by Theorem~\ref{thm:corsub} is a submatrix of $M^ZG(t)(M^T)^{\tensor(n-1)}$, so $M^Z$ has rank at least $2^K$. But $M^Z$ is a $2^K\times 2^K$ matrix, so $M^Z$ is invertible.
\end{proof}

Following the notation of \cite{caifu}, now denote the column vector $(M^Z)^{\tensor n}G$ of dimension $2^{Kn}$ by $\underline{G}^{*\leftarrow Z}$ and the column vector $(M^Z)^{\tensor(t-1)}\tensor M\tensor (M^Z)^{\tensor(n-t)}\cdot G$ of dimension $2^{Kn+\ell-K}$ by $\underline{G}^{t^c\leftarrow Z}$. Because $M^Z$ and $G(t)$ both have rank $2^K$, $\underline{G}^{*\leftarrow Z}$ and $\underline{G}^{t^c\leftarrow Z}$ also have rank $2^K$. We check that these can be realized as standard signatures.

\begin{lem}
	$\underline{G}^{*\leftarrow Z}$ is the standard signature of a generator matchgate of arity $Kn$.\label{lem:*}
\end{lem}

\begin{proof}
	Take the matchgate $G$, and in each block, attach an edge of weight 1 to external node $i$ ($1\le i\le \ell$) if $s_i = 1$. In the matchgate $G'$ we get from these operations, designate external nodes $p_1,...,p_K$ in each block as the new external nodes of $G'$. The resulting matchgate realizes $\underline{G}^{*\leftarrow Z}$.
\end{proof}

\begin{lem}
	$\underline{G}^{t^c\leftarrow Z}$ is the standard signature of a generator matchgate of arity $Kn-K+\ell$.\label{lem:tc}
\end{lem}

\begin{proof}
	The proof of Lemma~\ref{lem:tc} is almost identical to that of Lemma~\ref{lem:*}, except block $t$ is treated differently. Take the matchgate $G$, and in each block except the $t$-th one, attach an edge of weight 1 to external node $i$ ($1\le i\le \ell$) if $s_i = 1$. In the matchgate $G'$ we get from these operations, take the external nodes to be all $\ell$ external nodes in block $t$, as well as nodes $p_1,...,p_K$ in every other block. The resulting matchgate realizes $\underline{G}^{t^c\leftarrow Z}$.
\end{proof}

Now define $T = M(M^Z)^{-1}$. Here is the key step of the collapse theorem, making use of Theorem~\ref{thm:main_step}.

\begin{lem}
	$T$ is the standard signature of a $K$-input, $\ell$-output transducer.\label{lem:transducer}
\end{lem}

\begin{proof}
	We first express $T$ in terms of $\underline{G}^{*\leftarrow Z}$ and $\underline{G}^{t^c\leftarrow Z}$. If the entries of $\underline{G}^{t^c\leftarrow Z}$ are indexed by $(i_{1,1}\cdots i_{1,K})\cdots(i_{t-1,1}\cdots i_{t-1,K})(i'_1\cdots i'_{\ell})(i_{t+1,1}\cdots i_{t+1,\ell})\cdots(i_{n,1}\cdots i_{n,K})$, denote by $\underline{G}^{t^c\leftarrow Z}(t)$ the matrix form of $\underline{G}^{t^c\leftarrow Z}$ in which the rows are indexed by $i'_1\cdots i'_{\ell}$ and the columns are indexed by $(i_{1,1}\cdots i_{1,K})\cdots(i_{t-1,1}\cdots i_{t-1,K})(i_{t+1,1}\cdots i_{t+1,\ell})\cdots(i_{n,1}\cdots i_{n,K})$.

	Observe that $$\underline{G} = M^{\tensor n}G = T^{\tensor n}(M^Z)^{\tensor n}G = T^{\tensor n}\underline{G}^{*\leftarrow Z}$$ so that \begin{equation}\underline{G}^{t^c\leftarrow Z} = (T^Z)^{\tensor(t-1)}\tensor T\tensor(T^Z)^{\tensor(n-t)}\underline{G}^{*\leftarrow Z}.\label{eq:bothsides}\end{equation} Putting both sides of \eqref{eq:bothsides} in matrix form, we conclude that \begin{equation}\underline{G}^{t^c\leftarrow Z}(t) = T\underline{G}^{*\leftarrow Z}(t).\label{eq:t-relation}\end{equation}

	Applying Theorem~\ref{thm:main_step} to the arity-$Kn$ standard signature $\underline{G}^{*\leftarrow Z}$, we have a recognizer whose standard signature $\underline{R}$ satisfies $\underline{G}^{*\leftarrow Z}(t)\underline{R}(t) = I_{2^K}$. Right-multiplying both sides of \eqref{eq:t-relation} by $\underline{R}(t)$, we find that $$\underline{G}^{t^c\leftarrow Z}(t)\underline{R}(t) = T.$$

	Say that the generator realizing $\underline{G}^{t^c\leftarrow Z}$ as a standard signature has external nodes $X_{i,1}$, $X_{i,2}$, ..., $X_{i,K}$ in block $i$ for each $i\neq t$, and external nodes $Y_{t,1}$ ,..., $Y_{t,\ell}$ in block $t$. Say that the generator realizing $\underline{R}$ as a standard signature has external nodes $Z_{i,1},...,Z_{i,K}$ in each block $i$.

	Construct the transducer $\Gamma$ realizing $T$ as a standard signature by connecting $X_{i,j}$ with $Z_{i,j}$ for all $i\neq t$, $j\in[K]$. Designate $Y_{t,1},...,Y_{t,\ell}$ to be the output nodes of $\Gamma$ and $Z_{t,1}$, ..., $Z_{t,K}$ to be the input nodes of $\Gamma$.
\end{proof}

From Theorem~\ref{lem:transducer} we obtain the collapse theorem for domain size $2^K$.

\begin{thm}
	Any holographic algorithm on a basis of size $\ell$ and domain size $2^K$ which uses at least one generator signature of full rank can be simulated on a basis of size $K$.
\end{thm}

\begin{proof}
	Suppose the holographic algorithm in question uses signatures $R_i,G_j$ ($1\le i\le r$, $1\le j\le g$) defined by $\underline{R}_iM^{\tensor m_i} = R_i$ and $\underline{G}_j = M^{\tensor n_j}G_j$ over basis $M$. Say that $G_1$ has full rank, and let $Z = s+\{e_{p_1},...,e_{p_K}\}$ denote the full-rank $(K,\ell)$-cluster of rows in $G_1$ which must exist by Theorem~\ref{thm:sub}. By Lemma~\ref{lem:transducer}, $T:=M(M^Z)^{-1}$ is the standard signature of some transducer matchgate $\Gamma$. Let $\underline{R}'_i = \underline{R}_iT^{\tensor m_i}$ and $\underline{G}'_j = \underline{G}^{*\leftarrow Z}_j$; by Lemma~\ref{lem:triv_transducer}, $\underline{R}'_i$ is the standard signature of some recognizer, and by Lemma~\ref{lem:tc}, $\underline{G}'_j$ is the standard signature of some generator. We conclude that the $R_i,G_j$ can be simultaneously realized on the basis $M^Z$ of size $K$.
\end{proof}

\section{Acknowledgments}
I would like to thank Professor Leslie Valiant for patiently advising me throughout the course of this project and providing much valuable feedback as well as suggestions on papers to read. I am also extremely indebted to Professor Jin-Yi Cai for his incredibly extensive and important comments on drafts of this paper.


\appendix

\section{Planarizing Matchgates}
\label{app:crossover}

In the proof of Lemma~\ref{lem:transform}, we made several initial constructions of transducers to achieve certain row and column operations but noted that those constructions, specifically those shown in Figures~\ref{fig:2}, \ref{fig:3}, \ref{fig:4}, needed to be modified because they were not planar. Following the technique of Cai and Gorenstein \cite{caigoren}, we planarize those matchgates by replacing every edge crossing with the so-called \emph{crossover gadget} $X$ shown in Figure~\ref{fig:crossover}.

\begin{figure}[h]
\centering
	\includegraphics[width=0.75in]{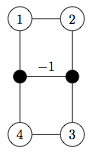}
	\caption{Crossover gadget. Unlabeled edges are of weight 1; labeled vertices are external nodes. Figure from \cite{caigoren}.}
	\label{fig:crossover}
\end{figure}

Because the standard signature $\underline{X}$ of the crossover gadget is given by $\underline{X}^{0000} = 1$, $\underline{X}^{0101} = 1$, $\underline{X}^{1010} = 1$, $\underline{X}^{1111} = -1$, and $\underline{X}^{\sigma} = 0$ for all other $\sigma\in\{0,1\}^4$ so that the standard signature remains invariant under any cyclic permutation of the external nodes, the orientation of the copy of $X$ placed over an edge crossing does not matter.

We first make precise our operation of planarizing matchgates, following the terminology of \cite{caigoren}. If an edge $\{u,v\}$ of weight $w$ crosses $t$ other edges, replace each of the $t$ crossings by a crossover gadget and replace the edge by $t+1$ edges connecting adjacent crossover gadgets. Of these $t+1$ edges, assign $t$ of them to have weight 1 and the remaining one to have weight $w$. Call the union of the $t+1$ edges the \emph{$u$-$v$-passage}.

Given a non-planar matchgate $\Gamma$, denote the matchgate obtained from planarizing $\Gamma$ by $\Gamma'$.

\begin{obs}
	Let $M$ be a perfect matching of $\Gamma'$ whose contribution $c$ to $\PerfMatch(\Gamma')$ is nonzero, and let $M'\subset M$ denote the edges not belonging to crossover gadgets. Then $M'$ is the union of $u$-$v$-passages corresponds to a perfect matching of $\Gamma$ whose contribution to $\PerfMatch(\Gamma)$ is $\pm c$.\label{obs:pm}
\end{obs}

\begin{proof}
	If an edge incident to one of the external nodes, say node 1, of a crossover gadget is present in $M$, then the edge incident to node 3 of the crossover gadget must be present in $M$ as well, as $\underline{X}^{\sigma} = 0$ if $\sigma_1\neq\sigma_3$. We conclude that $M$ is a union of $u$-$v$-passages. The corresponding perfect matching of $\Gamma$ has contribution $\pm c$ because each of the nonzero entries of $\underline{X}$ is $\pm 1$.
\end{proof}

We need to verify that the entries of $\underline{\Gamma}$ and $\underline{\Gamma'}$ are equal except for a select number of entries which differ by a factor of $-1$.

\begin{lem}
	Let $\Gamma$ be the $K$-input, $K$-output transducer shown in Figure~\ref{fig:2} with signature $L^{j,k}_3$. There exists a planar matchgate whose standard signature agrees with $\underline{\Gamma}$ on the main diagonal entries and entry $(1^{K}\oplus e_j\oplus e_k, 1^{K})$, and agrees with $\underline{\Gamma}$ everywhere else up to sign.
	\label{lem:crossover1}
\end{lem}

\begin{proof}
	Take the desired matchgate to be $\Gamma'$. Note that every subgraph of $\Gamma$ has at most one perfect matching. In other words, each entry of $\underline{\Gamma}$ arises from at most a single perfect matching. Therefore, by Observation~\ref{obs:pm}, $\underline{\Gamma}$ and $\underline{\Gamma'}$ agree everywhere up to sign. Now consider any main diagonal entry $\underline{\Gamma'}^{\sigma}_{\sigma} = \PerfMatch(\Gamma'\backslash Z)$. If $M$ is a perfect matching of $\Gamma'\backslash Z$ making a nonzero contribution to $\PerfMatch(\Gamma'\backslash Z)$, it corresponds to a perfect matching of $\Gamma\backslash Z$ making a nonzero contribution to $\PerfMatch(\Gamma\backslash Z)$. But the only such perfect matching does not contain the edge between left node $j$ to and left node $k$. Thus, the contribution of this matching and that of $M$ are both equal to 1.

	If $\underline{\Gamma}$ and $\underline{\Gamma'}$ disagree on entry $(1^{K}\oplus e_j\oplus e_k, 1^{K})$, modify $\Gamma$ by multiplying the weight of the edge connecting left nodes $j$ and $k$ by $-1$, and take the desired matchgate to be the corresponding $\Gamma'$.
\end{proof}

\begin{lem}
	Let $\Gamma$ be the $K$-input, $K$-output transducer shown in Figure~\ref{fig:3} with signature $L_4$. There exists a planar matchgate whose standard signature agrees with $\underline{\Gamma}$ everywhere up to sign.
	\label{lem:crossover2}
\end{lem}

\begin{proof}
	Take the desired matchgate to be $\Gamma'$. As in the proof of Lemma~\ref{lem:crossover1}, every subgraph of $\Gamma$ has at most one perfect matching, so we already know $\underline{\Gamma}$ and $\underline{\Gamma'}$ agree everywhere up to sign. 
\end{proof}

\begin{lem}
	Let $\Gamma$ be the $K$-input, $K$-output transducer shown in Figure~\ref{fig:4} with signature $L^{j}_5$. There exists a planar matchgate whose standard signature agrees with $\underline{\Gamma}$ on the main diagonal entries and entry $(1^{K}\oplus e_j, 1^{K}\oplus e_{q_{i+1}})$, and agrees with $\underline{\Gamma}$ everywhere else up to sign.
	\label{lem:crossover3}
\end{lem}

\begin{proof}
	Take the desired matchgate to be $\Gamma'$. As in the proof of Lemma~\ref{lem:crossover1}, every subgraph of $\Gamma$ has at most one perfect matching and $\underline{\Gamma}$ and $\underline{\Gamma'}$ agree on the main diagonal entries. If they disagree on entry $(1^{K}\oplus e_j, 1^{K}\oplus e_{q_{i+1}})$, modify $\Gamma$ by multiplying the weight of the edge connecting left node $j$ to right node $i+1$ by $-1$, and take the desired matchgate to be the corresponding $\Gamma'$.
\end{proof}

\end{document}